\documentclass[pra,aps,superscriptaddress,twocolumn]{revtex4}

\usepackage{placeins}
\usepackage{amsfonts, amsmath,amssymb,amsthm,bm}
\usepackage{mathrsfs}
\usepackage{graphicx}
\usepackage[usenames,dvipsnames]{color}
\usepackage{palatino}
\usepackage{fancyhdr}
\usepackage{mathrsfs}
\usepackage{multirow}
\usepackage{hyperref}
\usepackage{bigints}
\usepackage{soul}

\usepackage{framed}
\definecolor{shadecolor}{rgb}{.8,.8,.9}

\usepackage{braket}

\hypersetup{
   colorlinks=true,
   linkcolor=blue,
   citecolor=blue,
   urlcolor=Violet
}

\newcommand{\ca}{\mathcal A}

\newcommand{\cg}{\mathcal G}
\newcommand{\ch}{\mathcal H}
\newcommand{\calr}{\mathcal R}
\newcommand{\cu}{\mathcal U}
\newcommand{\ct}{\mathcal T}

\newcommand{\tr}{\mathrm{tr}}
\newcommand{\R}{\mathbb{R}}

\def\nn{\nonumber}
\def\q{{\quad}}

\def\be{\begin{equation}}
\def\ee{\end{equation}}
\def\ba{\begin{eqnarray}}
\def\ea{\end{eqnarray}}

\theoremstyle{plain}
\newtheorem{theorem}{Theorem}
\newtheorem{corollary}[theorem]{Corollary}
\newtheorem{lemma}[theorem]{Lemma}

\newtheorem{example}[theorem]{Example}

\newtheorem{definition}[theorem]{Definition}

\theoremstyle{definition}

\begin{document}
\title{Internal quantum reference frames for finite Abelian groups}
\author{Philipp A.\ H\"ohn}
\email{philipp.hoehn@oist.jp}
\affiliation{Okinawa Institute of Science and Technology Graduate University, Onna, Okinawa 904 0495, Japan}
\affiliation{Department of Physics and Astronomy, University College London, London, United Kingdom}
\author{Marius Krumm}
\email{marius.krumm@univie.ac.at}
\affiliation{Institute for Quantum Optics and Quantum Information, Austrian Academy of Sciences, Boltzmanngasse 3, A-1090 Vienna, Austria}
\affiliation{Vienna Center for Quantum Science and Technology (VCQ), Faculty of Physics, University of Vienna, Vienna, Austria}
\author{Markus P.\ M\"uller}
\email{markus.mueller@oeaw.ac.at}
\affiliation{Institute for Quantum Optics and Quantum Information, Austrian Academy of Sciences, Boltzmanngasse 3, A-1090 Vienna, Austria}
\affiliation{Vienna Center for Quantum Science and Technology (VCQ), Faculty of Physics, University of Vienna, Vienna, Austria}
\affiliation{Perimeter Institute for Theoretical Physics, 31 Caroline Street North, Waterloo ON N2L 2Y5, Canada}

\begin{abstract}
Employing internal quantum systems as reference frames is a crucial concept in quantum gravity, gauge theories and quantum foundations whenever external relata are unavailable. In this work, we give a comprehensive and self-contained treatment of such quantum reference frames (QRFs) for the case when the underlying configuration space is a finite Abelian group, significantly extending our previous work (M.\ Krumm, P.\ A.\ H\"ohn, and M.\ P.\ M\"uller, \emph{Quantum reference frame transformations as symmetries and the paradox of the third particle}, Quantum \textbf{5}, 530 (2021)). The simplicity of this setup admits a fully rigorous quantum information-theoretic analysis, while maintaining sufficient structure for exploring many of the conceptual and structural questions also pertinent to more complicated setups. We exploit this to derive several important structures of constraint quantization with quantum information-theoretic methods and to reveal the relation between different approaches to QRF covariance. In particular, we characterize the ``physical Hilbert space''\,---\,the arena of the ``perspective-neutral'' approach\,---\,as the maximal subspace that admits frame-independent descriptions of purifications of states. We then demonstrate the kinematical equivalence and, surprising, dynamical inequivalence of the ``perspective-neutral'' and the ``alignability'' approach to QRFs. While the former admits unitaries generating transitions between arbitrary subsystem relations, the latter, remarkably, admits \emph{no} such dynamics when requiring symmetry-preservation. We illustrate these findings by example of interacting discrete particles, including how dynamics can be described ``relative to one of the subsystems''.
\end{abstract}

\date{November 28, 2022}

\maketitle

\section{Introduction}
While reference frames are ubiquitous in physics, our archetypical picture of such ``rods and clocks'' is still painted in terms of classical physical objects. However, all physics is ultimately quantum, and thus it is crucial to realize that frames of reference are ultimately quantum systems, too. The consequences of this fundamental insight permeate several areas of physics, including quantum gravity~\cite{Rovellibook,Rovelli1,Rovelli2,Rovelli3,Dittrich1,Dittrich2,Thiemann,Tambornino}, quantum thermodynamics~\cite{LJR,LKJR,Aberg,LostaglioMueller,MarvianSpekkens,Erker,Cwiklinski,Woods1,Woods2}, quantum information theory~\cite{Bartlett,Marvian,Frameness,Modes,Smith2019,ResourceTheoryQRF,Palmer,Smith2016}, and the foundations of quantum physics~\cite{Aharonov1,Aharonov2,Aharonov3,Wigner,Araki,Yanase,Loveridge2017,Loveridge2018,Miyadera,Loveridge2020,HoehnMueller}.

In this work, we are concerned with an \emph{internal} or \emph{structural} notion of quantum reference frames (QRFs) as considered, for example, in Refs.~\cite{Giacomini,Vanrietvelde,Hamette,Vanrietvelde2,Hoehn:2018aqt,Hoehn:2018whn,Hoehn:2019owq,Hoehn:2020epv,Castro,Chataignier,Chataignier2,Chataignier3,Giacomini-spin1,Giacomini-spin2,Angelo,Hoehn:2021wet,Giacomini:2021gei,Ballesteros:2020lgl,Mikusch:2021kro,Baumann:2021ifs,Savi:2020qdl,Guerin:2018fja}. Not only does this notion of QRFs acknowledge that reference frames are quantum, but it also admits the possibility to transform between different QRFs that are, for example, relative to each other in superposition states. This perspective is natural if, as for example in some scenarios in quantum gravity and gauge theories, or the Page-Wootters mechanism~\cite{Page,Giovanetti,Alex1,Alex3}, a distinguished external reference frame may be unavailable, which implies that one has to choose a QRF among the internal quantum subsystems of which there may be many.

\begin{figure}[hbt]
\begin{center}
\includegraphics[width=.5\textwidth]{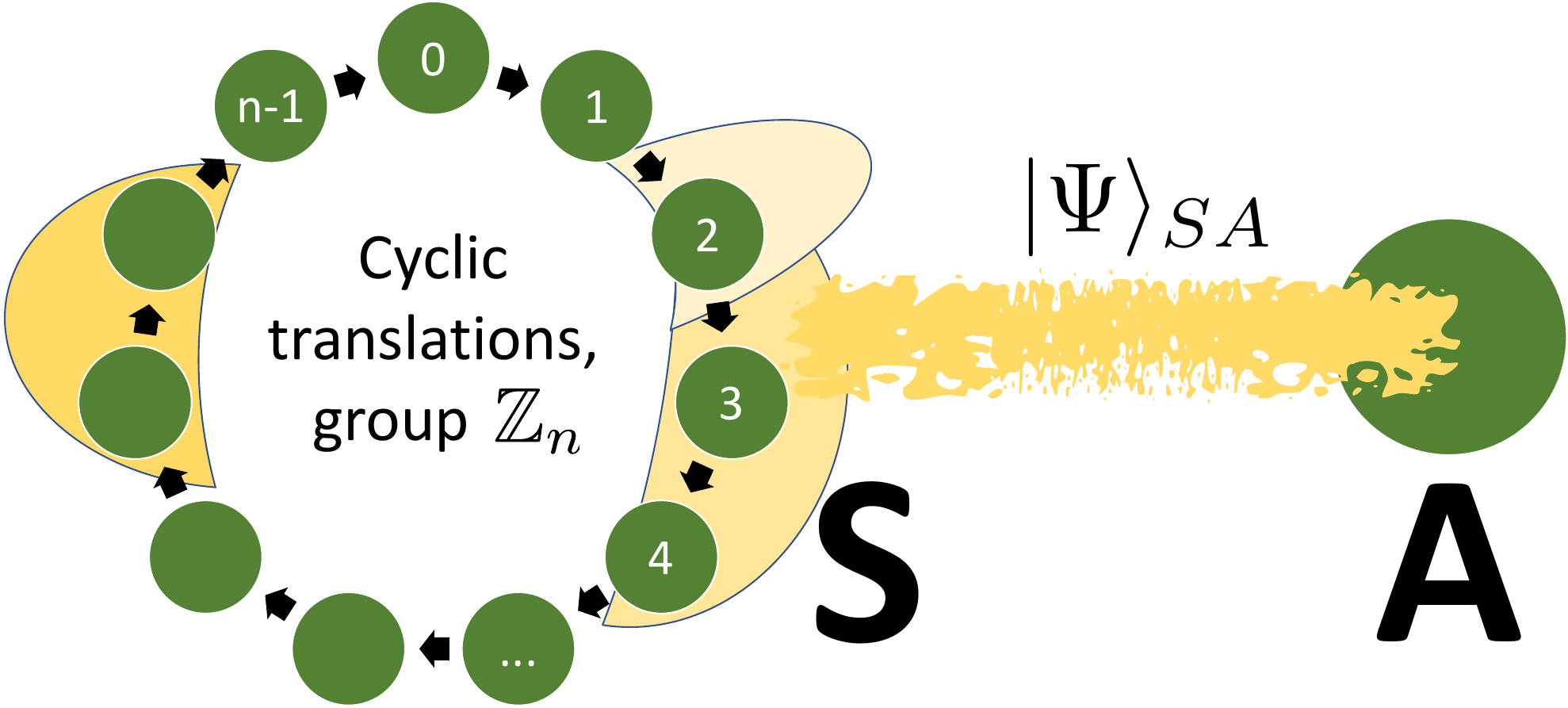}
\caption{As in Ref.~\cite{QRF1}, we consider a finite Abelian group $\cg$ (in this picture $\cg=\mathbb{Z}_n$) that plays a double role: first, as a classical configuration space that serves as a distinguished basis of a one-particle Hilbert space $\ch=\mathbb{C}(\cg)$: second, as a group of ``translations'' that acts on this space. We consider $N$ distinguishable particles in some quantum state $\rho_S\in\mathcal{L}(\ch^{\otimes N})$, potentially entangled with an unspecified purifying ancillary system $A$. Based on simple postulates, we study the resulting quantum symmetries and properties of such ``$\cg$-systems'' $S$.}
\label{fig_cyclic}
\end{center}
\end{figure}

The resulting genuine quantum symmetries~\cite{QRF1} lead to novel phenomena like the frame-dependence of superposition and entanglement~\cite{Giacomini,QRF1,Vanrietvelde,Hoehn:2019owq,Castro,Hamette,Hoehn:2021wet} or the very notion of subsystem~\cite{Hoehn:2021wet}. One specific goal, realized e.g.\ for relativistic spin~\cite{Giacomini-spin1}, is that this larger class of QRF transformations allows us to simplify the physical description: if we do not know how to handle the quantum case, let us perform a QRF transformation that renders some of the subsystems classical. This is arguably the main idea that underlies recent attempts to formulate a quantum version of Einstein's equivalence principle~\cite{Hardy1,Hardy2,GiacominiBrukner}.

Given the ambitious hopes associated to this notion of QRFs, and also the technical subtleties involved in the case of continuous symmetries (see e.g.\ Refs.~\cite{Bartlett,Loveridge2018}), it is crucial to study QRFs in special cases that admit a technically simple, mathematically rigorous treatment which allows us to expose structural and conceptual aspects in complete clarity. The case of \emph{finite Abelian groups} is arguably the best candidate for this: its finite-dimensionality admits a rigorous quantum information-theoretic treatment, and explicit diagonalization is possible via the discrete Fourier transform. However, this special case is already complex enough to encompass interesting physical scenarios like the ``paradox of the third particle''~\cite{Angelo,QRF1}. Moreover, as we will see below, it is rich enough to demonstrate important structures of constraint quantization which appears in canonical approaches to quantum gravity and gauge theories.

In this paper, we give a comprehensive treatment of internal QRFs for finite Abelian groups. This extends and generalizes results of our previous publication~\cite{QRF1}, but adds many important novel results and insights. In particular, and in contrast to Ref.~\cite{QRF1}, we explicitly demonstrate kinematical equivalence of the ``perspective-neutral'' approach to QRFs~\cite{Vanrietvelde,Vanrietvelde2,Hoehn:2018aqt,Hoehn:2018whn,Hoehn:2019owq,Hoehn:2020epv,Hoehn:2021wet,Giacomini:2021gei}, which invokes constraint quantization, with the one based on ``aligning states to the QRF''~\cite{Hamette,QRF1} in this context, which in those references, in turn, was shown to be equivalent to the QRF approach developed in Refs.~\cite{Giacomini,Giacomini-spin1,Giacomini-spin2,Ballesteros:2020lgl,Mikusch:2021kro}. Moreover, we classify the symmetry-preserving dynamics of systems subject to such quantum symmetries, and use this to demonstrate the dynamical inequivalence of these two approaches. Despite relying on some of the results of Ref.~\cite{QRF1}, the presentation is fairly self-contained. Readers may want to consult Ref.~\cite{QRF1} occasionally for the proofs of some statements or for more details about the paradox of the third particle.\\

\textbf{This paper is organized as follows.} In Section~\ref{SecKinematics}, we recapitulate, consolidate and generalize our results of Ref.~\cite{QRF1}: we introduce the notion of a ``$\cg$-system'' as a quantum system that holds a finite Abelian group as its configuration space. We show how QRF transformations appear as the symmetries of such systems, analyze the structure of those transformations, and determine two classes of symmetric observables: ``invariant'' and ``relational'' observables, related to incoherent resp.\ coherent group averaging. We show that these encode the subalgebras of observables that are measurable without access to an external reference frame, where the relational observables are further characterized as those measurable on the maximal subspace of states of the $\cg$-system that can be purified with an ancilla in an external-frame-independent manner. We introduce and contrast two different ways to describe relational quantum physics: via \emph{alignable states} and via \emph{relational states}, where the latter are the states underlying the ``perspective-neutral'' approach. We also summarize the notion of ``relational trace'' introduced in Ref.~\cite{QRF1}, which generalizes the partial trace to $\cg$-systems.

In Section~\ref{SecRelAl}, we prove that relational and alignable states are kinematically equivalent. In Section~\ref{Section:Time}, we classify the possible dynamics that respects a $\cg$-system's quantum symmetries. We show that relational states admit a much larger class of such symmetry-preserving dynamics than alignable states, and illustrate this by example of $N$ interacting particles on the discrete circle. Specifically, alignable states do not admit \emph{any} symmetry preserving unitary dynamics that can generate superpositions of subsystem relations. Finally, we conclude in Section~\ref{SecConclusions}.

\section{Finite Abelian groups: kinematics}
\label{SecKinematics}
We consider physical systems that can be interpreted as ``quantizations of symmetric classical systems'' in a certain sense --- in fact, we will see later that our scenario and its mathematical structures are closely related to the quantization of gauge theories. Our starting point is a finite Abelian group $\mathcal{G}$ that we interpret as a classical configuration space. For example, the cyclic group $\mathbb{Z}_n=\{0,1,2,\ldots,n-1\}$ with addition modulo $n$ can be interpreted as a set of $n$ equiangular points on a circle, see Figure~\ref{fig_cyclic}. Imagine $N$ distinguishable particles that can be placed somewhere on this discrete circle, where we allow that different particles may occupy the same place. A configuration of the $N$ particles is then given by an $N$-tuple $\mathbf{g}=(g_1,g_2,\ldots,g_N)$ with $g_i\in\mathcal{G}$.

We will now postulate that ``physics is translation-invariant'' in a natural, but very specific sense. Think of an observer (say, Alice) that is presented with a given $N$-particle configuration. Under translation-invariance, the observer will not be able to determine which of the positions represents the origin, $0$. Alice may arbitrarily declare one of the positions as $0$ and describe the configuration by $\mathbf{g}$; but then another observer (say, Bob) may choose a different origin and describe the configuration instead by $\mathbf{g'}=g\mathbf{g}:=(gg_1,\ldots,gg_N)$, where $g\in\mathcal{G}$ denotes the group element that translates Bob's choice of origin to Alice's.

While Alice and Bob may disagree on the location of the origin (and other facts), there are also some structures on which they agree. In fact, \emph{postulating the structures on which they agree will serve as our definition of the physical symmetry}. We will make the following two natural assumptions:
\begin{itemize}
	\item[(i)] For any given \textbf{single configuration}, Alice and Bob agree on the pairwise \textbf{relations} between the particles. That is, if Alice chooses description $\mathbf{g}$, then Bob must choose some description $\mathbf{g}'$, where $g_i^{-1}g_j={g'}_i^{-1}g'_j$ for all $i,j\in\{1,\ldots,N\}$.
\end{itemize}
Let us introduce some notation to formalize this assumption which will be useful in the following. We write
\[
   (g,\mathbf{h}g)=(g,h_1 g,h_2 g,\ldots, h_{N-1}g)\qquad(g\in\mathcal{G},\mathbf{h}\in\mathcal{G}^{N-1}),
\]
and note that $\mathbf{h}$ encodes all pairwise relations by encoding all relations to the first particle. Namely, $g_i=h_{i-1}g_1$ (with the convention $h_0=e$), and thus $g_j=h_{j-1}h_{i-1}^{-1}g_i$. Then the assumption says that if Alice chooses description $(g,\mathbf{h}g)$, then Bob must choose $(g',\mathbf{h}g')$ for some $g'\in\mathcal{G}$.

Furthermore, we assume:
\begin{itemize}
	\item[(ii)] For any given \textbf{pair of} two possible \textbf{configurations}, Alice and Bob agree on whether one is a global translation of the other (and which one) or not.
\end{itemize}
Think of two possible preparation procedures that result in configurations which Alice would describe by $\mathbf{g}$ and $\mathbf{j}$, respectively. Bob will in general choose two different descriptions, $\mathbf{g}'$ and $\mathbf{j}'$. However, if $\mathbf{j}$ is a global translation of $\mathbf{g}$, i.e.\ if there is some $g\in\mathcal{G}$ such that
\[
   \mathbf{j}=g\mathbf{g},\mbox{ i.e.\ }(j_1,j_2,\ldots,j_N)=(gg_1,gg_2,\ldots,gg_N),
\]
then we assume that Bob agrees on this fact: his two descriptions are then \emph{also} related by
\[
   \mathbf{j}'=g\mathbf{g}',\mbox{ i.e.\ }(j'_1,j'_2,\ldots,j'_N)=(gg'_1,gg'_2,\ldots,gg'_N).
\]
Since we are not assuming more than this, our assumptions encode a specific version of \emph{background-independence}, which we now illustrate by an example. Suppose that $N=3$ (and $\mathcal{G}=\mathbb{Z}_n$ for $n\geq 8$), and Alice describes two possible configurations as
\[
   \mathbf{g}=(0,1,2),\quad \mathbf{j}=(0,1,3).
\]
One possible choice of descriptions for Bob, consistent with our assumptions, is given by
\[
   \mathbf{g}'=(0,1,2),\quad \mathbf{j}'=(4,5,7).
\]
At first sight, this is counterintuitive. Intuitively, one would perhaps think of a ``choice of description'' as concretely happening in the following way: Alice is presented with a certain $N$-particle configuration, and she makes a choice to describe it by $\mathbf{g}$. This, in particular, implies a certain choice of origin, $e=0$. Alice now plants a little flag at the $(e=0)$-site of the discrete circle and fixes the physical origin once and for all. Then, in the second run of the experiment, when she is presented another configuration, her little flag will break translation-invariance and allow her to give a unique description $\mathbf{j}$ of the configuration.

But if this was the case, then the above choices of descriptions would be inconsistent: $\mathbf{g}=\mathbf{g}'$ would then imply $\mathbf{j}=\mathbf{j}'$. Hence, our two assumptions (together with our choice of \emph{not postulating any additional assumptions}) imply that ``planting a flag'' in this form is impossible. In other words, \emph{there is no background structure that allow us to identify configurations across different runs of the experiment}. We can also interpret $\mathbf{g}$ and $\mathbf{j}$ as different \emph{modalities}, i.e.\ as different possibilities of what could, in any single implementation, actually be the case. Then our scenario is constructed such that the different modalities (``possible worlds'') can be translated independently from one another. In the quantum case, this will then also apply to different branches of the wave function.

Let us turn to the quantum case. Now the $N$-particle configurations correspond to orthonormal vectors of a Hilbert space $\mathcal{H}^{\otimes N}$, where $\mathcal{H}={\rm span}\{|g\rangle\,\,|\,\,g\in\mathcal{G}\}=\mathbb{C}(\mathcal{G})$ is the Hilbert space of complex functions on the group. Alice and Bob will now describe \emph{quantum states} instead of classical configurations, and their different choices of description are related by a unitary $U$. Since we think of this as a quantization of the classical system, we make a third and final assumption:
\begin{itemize}
	\item[(iii)] Alice and Bob agree on the \textbf{set of basis vectors} $\{|\mathbf{g}\rangle\,\,|\,\,\mathbf{g}\in\mathcal{G}^N\}$, but not in general on the labelling of these basis vectors, except to the extent described by assumptions (i) and (ii).
\end{itemize}

\subsection{$\mathcal{G}$-systems and their symmetries}
The three assumptions above lead us to the following definition. We use the notation $U_g$ for the $g$-translation, i.e.\ $U_g|j\rangle=|gj\rangle$ for $j\in\mathcal{G}$.
\begin{definition}[$\mathcal{G}$-system]
Given some finite Abelian group $\mathcal{G}\neq\{\mathbf{1}\}$, a \emph{$\mathcal{G}$-system} is	 a quantum system described by a (kinematical) Hilbert space $\mathcal{H}^ {\otimes N}$, where $\mathcal{H}=\mathbb{C}(\mathcal{G})$. It carries a distinguished orthonormal basis 
\[
   \{|\mathbf{g}\rangle=|g_1,g_2,\ldots,g_N\rangle\,\,|\,\,g_i\in\mathcal{G}\}.
\]
Every unitary transformation $U$ with the following properties is a \emph{symmetry} of the $\mathcal{G}$-system:
\begin{itemize}
	\item[1.] $U$ maps classical configurations to classical configurations, i.e.\ $U|\mathbf{g}\rangle=|\mathbf{g}'\rangle$.
	\item[2.] On the classical configurations, $U$ preserves all pairwise relations, i.e.\ $U|g,\mathbf{h}g\rangle=|g',\mathbf{h}g'\rangle$.
	\item[3.] If some classical configuration is a global $g$-translation of another one, then $U$ preserves this fact, i.e.
	\[
	   |\mathbf{g}\rangle=U_g^{\otimes N}|\mathbf{j}\rangle \Rightarrow U|\mathbf{g}\rangle=U_g^{\otimes N}\left(U|\mathbf{j}\rangle\right).
	\]
\end{itemize}
The group of symmetries $U$ is denoted $\mathcal{U}_{\rm sym}$.
\end{definition}
When we say that unitaries $U\in\mathcal{U}_{\rm sym}$ are \emph{symmetries} of the $\mathcal{G}$-system, then this amounts to the claim that any quantum state $\rho$ and its transformed version $U\rho U^\dagger$ are physically indistinguishable if the $\mathcal{G}$-system is considered in isolation, i.e.\ without any external relatum. We can certainly imagine that we modify the physical scenario in a way that breaks the symmetry --- for example, we can add another quantum system close to the origin $e$ of the $\mathcal{G}$-system and make it interact with it, such that the strength of interaction is larger for $\mathcal{G}$-system particles that are closer to the origin. Such a system would then serve as an external reference frame. In fact, we may argue that the very possibility of doing something like this is crucial: writing down a definition that distinguishes $|\mathbf{g}\rangle$ from $|g\mathbf{g}\rangle$, for example, assumes that there is in principle a matter of fact in the world that motivates this distinction to begin with.

The symmetry group can be described as follows:
\begin{lemma}[Ref.~\cite{QRF1}, Lemma 5]
The symmetry group of a $\mathcal{G}$-system is
\begin{equation}
   \mathcal{U}_{\rm sym}=\left\{\left. \bigoplus_{\mathbf{h}\in\mathcal{G}^{N-1}} U_{g(\mathbf{h})}^{\otimes N}\,\,\right|\,\, g(\mathbf{h})\in\mathcal{G}\right\}.
   \label{eqUSym}
\end{equation}
That is, the symmetries are the \emph{relation-conditional translations}: depending on the pairwise relations $\mathbf{h}$, a global translation by some $g(\mathbf{h})$ is implemented.
\end{lemma}

This symmetry group $\mathcal{U}_{\rm sym}$ is much larger than the group $\cg$ which has originally defined our Hilbert space. $\cu_{\rm sym}$ is the discrete analog of gauge transformations that depend on gauge-invariant field configurations in a field theory, while $\cg$ is the discrete analog of field-independent gauge transformations. Note that every choice of pairwise relations $\mathbf{h}$ spans a subspace $\mathcal{H}_{\mathbf{h}}:={\rm span}\{|g,\mathbf{h}g\rangle\,\,|\,\,g\in\mathcal{G}\}$, and the total Hilbert space is $\mathcal{H}^{\otimes N}=\bigoplus_{\mathbf{h}\in\mathcal{G}^{N-1}}\mathcal{H}_{\mathbf{h}}$. In Eq.~(\ref{eqUSym}), each of the $U_{g(\mathbf{h})}^{\otimes N}$ acts, by definition, only on the subspace $\ch_{\mathbf{h}}$.

We can already see at this point that these symmetry transformations implement versions of the QRF transformations that have been described in earlier publications, e.g.\ Refs.~\cite{Giacomini,Hamette}:
\begin{example}
Consider $N=3$ particles (labelled $A$, $B$ and $C$) on a discrete circle of $n\geq 4$ points, i.e.\ on a $\mathbb{Z}_n$-system. Set
\[
   g(1,2):=-2\equiv n-2,\qquad g(1,3):=-3 \equiv n-3,
\]
and define $g(\mathbf{h})$ arbitrarily for all other choices of $\mathbf{h}\in\mathbb{Z}_n^2$. Then $U:=\bigoplus_{\mathbf{h}}U_{g(\mathbf{h})}^{\otimes N}$ is a symmetry transformation. Consider the state
\[
   |\psi\rangle:=|0\rangle_A\otimes |1\rangle_B\otimes\frac{|2\rangle_C+|3\rangle_C}{\sqrt{2}}.
\]
Applying $U$ to it, we obtain
\[
   |\psi'\rangle:=U|\psi\rangle= \frac{|-2\rangle_A\otimes|-1\rangle_B+|-3\rangle_A\otimes|-2\rangle_B}{\sqrt{2}}\otimes |0\rangle_C.
\]
The quantum states $|\psi\rangle$ and $|\psi'\rangle$ are related by a symmetry transformation; hence they are physically indistinguishable without access to any external frame of reference. We can see $|\psi\rangle$ and $|\psi'\rangle$ as two possible descriptions of the same physical situation: $|\psi\rangle$ can be viewed as the description of the quantum state relative to $A$, and $|\psi'\rangle$ relative to $C$ (since these particles are located in the origin).

We also see that $|\psi\rangle$ is a product state while $|\psi'\rangle$ is entangled --- hence, at least formally, notions of correlation and entanglement become dependent on the choice of description.
\end{example}

\subsection{Invariant and relational observables}
Which observables can we measure internally on an isolated $\mathcal{G}$-system without access to an external symmetry-breaking degree of freedom? These are the observables that are invariant under all symmetry transformations:
\begin{definition}[Ref.~\cite{QRF1}, Definition 7]
We define the \emph{invariant subalgebra} as
\[
   \mathcal{A}_{\rm inv}:=\{A\in\mathcal{L}(\mathcal{H}^{\otimes N})\,\,|\,\, [U,A]=0\mbox{ for all }U\in\mathcal{U}_{\rm sym}\},
\]
where $\mathcal{L}(\ch^{\otimes N})$ denotes the linear maps (operators) on $\ch^{\otimes N}$. These are the operators that are invariant under all symmetry transformations $A\mapsto UAU^\dagger$. The self-adjoint elements $A=A^\dagger\in\mathcal{A}_{\rm inv}$ are called \emph{invariant observables}.
\end{definition}
To write down $\mathcal{A}_{\rm inv}$ explicitly, consider the \emph{characters}~\cite{Conrad} of the finite Abelian group $\mathcal{G}$. A character is a homomorphism $\chi:\mathcal{G}\to S^1$, where $S^1$ is the complex unit circle, i.e.\ $\chi(gh)=\chi(g)\chi(h)$ for all $g,h\in\mathcal{G}$. These are exactly the irreducible (and thus automatically one-dimensional) representations of $\cg$; according to different conventions of nomenclature, they are called \emph{irreducible characters}~\cite{Simon}. We denote the set of all characters of $\mathcal{G}$ by $\mathcal{\hat G}$, and we have $|\mathcal{G}|=|\mathcal{\hat G}|$. The well-known orthogonality relation $\sum_{g\in\mathcal{G}} \overline{\chi(g)}\chi'(g)=|\cg|\,\delta_{\chi,\chi'}$ implies that the states
\[
   |\mathbf{h};\chi\rangle:=\frac 1 {\sqrt{|\mathcal{G}|}} \sum_{g\in\mathcal{G}} \chi(g^{-1})|g,\mathbf{h}g\rangle\qquad(\mathbf{h}\in\mathcal{G}^{N-1},\chi\in\mathcal{\hat G})
\]
are an orthonormal basis of $\mathcal{H}^{\otimes N}$. They are eigenstates of the global translations:
\[
   U_g^{\otimes N}|\mathbf{h};\chi\rangle=\chi(g)|\mathbf{h};\chi\rangle \mbox{ for all }g\in\mathcal{G}.
\]
We can think of the $|\mathbf{h};\chi\rangle$ as discrete analogues of ``total momentum eigenstates''. This becomes most transparent in the special case of $\mathcal{G}=\mathbb{Z}_n$:
\begin{example}[Cyclic group; Ref.~\cite{QRF1}, Example 8]
\label{ExCyclic}
For $\mathcal{G}=\mathbb{Z}_n$, i.e.\ the discrete circle of Figure~\ref{fig_cyclic}, we have
\[
   \chi_k(g)=e^{\frac{2\pi i k g}{n}}\qquad(k\in\{0,1,\ldots,n-1\}),
\]
and the eigenstates are
\[
   |\mathbf{h};\chi_k\rangle=\frac 1 {\sqrt{n}}\sum_{g=0}^{n-1} e^{-\frac{2\pi i k g}n} |g,\mathbf{h}g\rangle,
\]
where $\mathbf{h}g=\mathbf{h}+g\mbox{ mod }n\equiv(h_1+g\mbox{ mod }n,\ldots,h_{N-1}+g\mbox{ mod }n)$. That is, the position eigenbasis $|g,\mathbf{h}g\rangle$ and the character (``momentum'') eigenbasis $|\mathbf{h};\chi_k\rangle$ are related by a discrete Fourier transform.
\end{example}
Before characterizing the invariant subalgebra, let us look at an even smaller subalgebra. Consider the \emph{relational} or \emph{physical} Hilbert space
\[
   \mathcal{H}_{\rm phys}:=\ch_{\mathbf{1}}={\rm span}\{|\mathbf{h};\mathbf{1}\rangle\,\,|\,\,\mathbf{h}\in\mathcal{G}^{N-1}\},
\]
where $\mathbf{1}\in\mathcal{\hat G}$ is the character with $\mathbf{1}(g)=1$ for all $g\in\mathcal{G}$. The choice of name comes from the fact that $\mathcal{H}_{\rm phys}$ is the discrete analog of the so-called physical Hilbert space of constraint quantization. As we shall see later, it is the arena of the ``perspective-neutral'' approach to quantum frame covariance. It is easy to see that $\mathcal{H}_{\rm phys}$ consists of all vectors that are invariant under $\mathcal{U}_{\rm sym}$,
\[
   \mathcal{H}_{\rm phys}=\{|\psi\rangle\in\mathcal{H}^{\otimes N}\,\,|\,\, U|\psi\rangle=|\psi\rangle\mbox{ for all }U\in\mathcal{U}_{\rm sym}\}.
\]
In Eq.~(\ref{eqPiPhys}) below, we will see that $\mathcal{H}_{\rm phys}$ can also be characterized as the set of vectors that are invariant under all global translations $U_g^{\otimes N}$, recovering the definition used in other works like, e.g.,~\cite{Hoehn:2019owq,all}.
Let us denote the subalgebra of operators that are fully supported on $\mathcal{H}_{\rm phys}$ by $\mathcal{A}_{\rm phys}$, the \emph{relational subalgebra}. Clearly, $\mathcal{A}_{\rm phys}\subset\mathcal{A}_{\rm inv}$, but the invariant subalgebra is strictly larger than the relational subalgebra:
\begin{lemma}[Ref.~\cite{QRF1}, Lemma 10]
The invariant subalgebra consists of the block matrices of the form
\[
   \mathcal{A}_{\rm inv}=\left\{ A_{\rm phys}\oplus \bigoplus_{\mathbf{h}\in\mathcal{G}^{N-1},\chi\neq\mathbf{1}} a_{\mathbf{h};\chi}|\mathbf{h};\chi\rangle\langle\mathbf{h};\chi|\right\},
\]
where $A_{\rm phys}\in\mathcal{A}_{\rm phys}$, and the $a_{\mathbf{h};\chi}$ are complex numbers.
\end{lemma}
Thus, our quantum symmetries impose an emergent superselection rule: we cannot measure observables with coherences between different characters $\chi$, or between different relations $\mathbf{h}$\,---\,unless $\chi=\mathbf{1}$, in which case we are dealing with relational states and observables.

Note that both $\ca_{\rm phys}$ and $\ca_{\rm inv}$ are (in general strict) subalgebras of $\ca'_{\rm inv}=\{A\,\,|\,\, [A,U_g^{\otimes N}]=0\}$, the algebra of observables invariant under all global translations~\cite{QRF1}, i.e.\ of observables block-diagonal in the characters $\chi$.

To understand the different roles of $\mathcal{A}_{\rm inv}$ and $\mathcal{A}_{\rm phys}$, let us look at the states. Invariant states $\rho\in\mathcal{A}_{\rm inv}$ are those that are independent of the choice of description, $U\rho U^\dagger=\rho$ for all $U\in\mathcal{U}_{\rm sym}$. Different observers who use different descriptions (and are thus related by symmetry transformations) will agree on the description of these states; in this sense, they become \emph{speakable information}. In fact, if $\sigma\in\mathcal{S}(\mathcal{H}^{\otimes N})$ is an arbitrary (not in general invariant) state, then
\[
   \tr(\sigma A)=\tr(\Pi_{\rm inv}(\sigma)A)\quad\mbox{for all }A\in\mathcal{A}_{\rm inv},
\]
where, according to~\cite[Lemma 11]{QRF1},
\begin{equation}
   \Pi_{\rm inv}(\sigma):=\frac 1 {|\mathcal{U}_{\rm sym}|}\sum_{U\in\mathcal{U}_{\rm sym}} U\sigma U^\dagger
   \label{eqPiInv}
\end{equation}
is the projection of $\rho$ into $\mathcal{A}_{\rm inv}$. Thus, if we take the perspective that only the invariant observables are physically meaningful, then the physically relevant content of any state is given by its projection into the invariant subalgebra.

Why should we be interested in the subspace $\ch_{\rm phys}$ of \emph{vectors} that are invariant under all $U$, as opposed to the set of all \emph{states} with that property, i.e.\ $\ca_{\rm inv}$? One possible motivation is that states on the physical subspace $\ch_{\rm phys}$ have a stronger symmetry property: not only are those states invariant under a change of description (since $\ca_{\rm phys}\subset\ca_{\rm inv}$), but also \emph{the quantum information that these states carry about other systems} is invariant under a change of description. To see this, consider some ancillary system $A$ with $\dim A\geq\dim S$ that purifies the state $\rho_S$ of the $\cg$-system $S$, i.e. $\rho_S={\rm Tr}_A |\Psi\rangle\langle\Psi|_{SA}$. As usual in quantum information theory, we regard $A$ as unavailable to the agent and leave it completely unspecified. After all, we want to preserve the quantum information held by $S$ about \emph{any} other quantum system $A$, which is why we will not make any further assumptions on $A$. In particular, $A$ may or may not carry any symmetries related to those of $S$.

Then we obtain the following characterization:
\begin{lemma}
\label{LemPurification}
 For all purifications $|\Psi\rangle_{SA}$ of every state $\rho_S\in\ca_{\rm phys}$, it holds
\[
   U_S\otimes\mathbf{1}_A|\Psi\rangle\langle\Psi|_{SA}U_S^\dagger\otimes\mathbf{1}_A=|\Psi\rangle\langle\Psi|_{SA}\enspace\forall \,U_S\in\mathcal{U}_{\rm sym}.
\]
Conversely, all mixed states $\rho_S$ with this property are elements of $\ca_{\rm phys}$.
\end{lemma}

\begin{proof}
Suppose that $\rho_S\in\ca_{\rm phys}$, then every purification can be written in the form
\[
   |\Psi\rangle_{SA}=\sum_i \sqrt{\lambda_i}|i\rangle_S \otimes |i\rangle_A,
\]
where the $\lambda_i$ are the eigenvalues of $\rho_S$, and $|i\rangle_S\in\ch_{\rm phys}$. Hence $U_S|i\rangle_S=|i\rangle_S$ for all $i$, and so $U_S\otimes\mathbf{1}_A|\Psi\rangle_{SA}=|\Psi\rangle_{SA}$.

On the other hand, suppose that $\rho_S$ is any mixed state for which all purifications $|\Psi\rangle_{SA}$ satisfy the statement of the lemma. Tracing over $A$, we find $U_S\rho_S U_S^\dagger=\rho_S$, i.e.\ $\rho_S\in\ca_{\rm inv}$. Thus, $\rho_S$ is of the form
\[
   \rho_S=\sum_{i=1}^d \lambda_i |i\rangle\langle i|+\sum_{\mathbf{h},\chi\neq\mathbf{1}} \lambda_{\mathbf{h};\chi}|\mathbf{h};\chi\rangle\langle\mathbf{h};\chi|,
\]
where $d=\dim\ch_{\rm phys}$, $\{|i\rangle\}$ is an orthonormal basis of $\mathcal{H}_{\rm phys}$, and the $\lambda_{\bullet}$ are all non-negative and sum to one. A particular purification of $\rho_S$ is
\[
   |\Psi\rangle_{SA}=\sum_{i=1}^d\sqrt{\lambda_i}|i\rangle_S\otimes |i\rangle_A+\sum_{\mathbf{h},\chi\neq\mathbf{1}} \sqrt{\lambda_{\mathbf{h};\chi}}|\mathbf{h};\chi\rangle_S\otimes |\mathbf{h};\chi\rangle_A.
\]
Set $U_S=\bigoplus_{\mathbf{h}}U_{g(\mathbf{h})}^{\otimes N}$, then $U_S\in\mathcal{U}_{\rm sym}$. For every choice of $g(\mathbf{h})$, we have $U_S\otimes\mathbf{1}_A|\Psi\rangle_{SA}=e^{i\theta}|\Psi\rangle_{SA}$ for some global phase $\theta$ that may depend on the choice of $g(\mathbf{h})$. We have
\begin{eqnarray*}
U_S\otimes\mathbf{1}_A|\Psi\rangle_{SA}&=&\sum_{i=1}^d \sqrt{\lambda_i} |i\rangle_S\otimes|i\rangle_A\\
&&+\sum_{\mathbf{h},\chi\neq\mathbf{1}}\sqrt{\lambda_{\mathbf{h};\chi}}\chi(g(\mathbf{h}))|\mathbf{h};\chi\rangle_S\otimes|\mathbf{h};\chi\rangle_A.
\end{eqnarray*}
Suppose that $\rho_S\not\in\ca_{\rm phys}$. First, consider the case that there is some $i$ with $\lambda_i\neq 0$. Then we must have at least one $\lambda_{\mathbf{h},\chi}\neq 0$. Since $\chi\neq\mathbf{1}$, there exists some $g$ with $\chi(g)\neq 1$. Set $g(\mathbf{h}):=g$, then $U_S$ changes the relative phases in the components of $|\Psi\rangle_{SA}$; this is a contradiction to our earlier claims.

Second, suppose that $\lambda_i=0$ for all $i$. Since $\rho_S$ is mixed, there must exist $(\mathbf{h},\chi)\neq (\mathbf{h}',\chi')$ such that $\lambda_{\mathbf{h};\chi}\neq 0$ and $\lambda_{\mathbf{h}';\chi'}\neq 0$. If $\chi=\chi'$, then $\mathbf{h}\neq\mathbf{h}'$. Since $\chi\neq\mathbf{1}$, we can choose $g(\mathbf{h})$ and $g(\mathbf{h}')$ such that $\chi(g(\mathbf{h}))\neq\chi(g(\mathbf{h}'))$. Again, this introduces relative phases into $|\Psi\rangle_{SA}$ which is a contradiction. Finally, if $\chi\neq\chi'$, then choose some $g$ with $\chi(g)\neq\chi'(g)$, and set $g(\mathbf{h})=g(\mathbf{h}')=g$. This also introduces relative phases into $|\Psi\rangle_{SA}$.
\end{proof}
This lemma underlines the physical significance of the relational subalgebra $\ca_{\rm phys}$. We will later see that there is an additional reason for preferring it over $\ca_{\rm inv}$: in contrast to the invariant subalgebra, it admits an invariant notion of partial trace. Furthermore, it will, in fact, be tomographically complete for the invariant information in the alignable states that we introduce below.

The orthogonal projection onto $\mathcal{H}_{\rm phys}$ can be written
\begin{equation}
   \Pi_{\rm phys}=\frac 1 {|\mathcal{U}_{\rm sym}|} \sum_{U\in\mathcal{U}_{\rm sym}}U=\frac 1 {|\mathcal{G}|} \sum_{g\in\cg}U_g^{\otimes N},
   \label{eqPiPhys}
\end{equation}
where the second equality follows from \cite[Lemma 11]{QRF1}.
Moreover, our symmetry considerations motivate us to define two notions of \emph{equivalence} of states.
\begin{definition}[Ref.~\cite{QRF1} Definition 13]
We call two states $\rho,\sigma\in\mathcal{S}(\ch^{\otimes N})$ \emph{symmetry-equivalent}, and write $\rho\simeq \sigma$, if there exists some $U\in\mathcal{U}_{\rm sym}$ with $\sigma=U\rho U^\dagger$. We call them \emph{observationally equivalent} and write $\rho\sim\sigma$ if $\tr(A\rho)=\tr(A\sigma)$ for all $A\in\ca_{\rm inv}$. The equivalence class of states $\sigma$ with $\sigma\sim\rho$ is denoted $[\rho]$.
\end{definition}
Clearly $\rho\simeq\sigma$ implies $\rho\sim\sigma$, but not vice versa. Furthermore, according to~\cite[Lemma 14]{QRF1}, $\rho\sim\sigma$ is equivalent to $\Pi_{\rm inv}(\rho)=\Pi_{\rm inv}(\sigma)$.

\subsection{Alignable states}
\label{SubsecAlignable}
So far, we have taken an operational perspective: Above, we have asked which states are distinguishable (and which observables measurable) by observers constrained by the quantum symmetries of a $\mathcal{G}$-system.  We have seen that observationally equivalent states (say, $\rho$ and $\sigma$ with $\rho\sim\sigma$) agree on all  predictions that can be tested by such observers. Thus, these observers can choose any state from the equivalence class $[\rho]$ of $\rho$ as a description of the corresponding preparation procedure. 

In Ref.~\cite{QRF1}, we have formulated a corresponding communication task. Two observers (Alice and Bob) obtain a description of $[\rho]$. They are not allowed to communicate, but they each have to write a representative $\sigma\in[\rho]$ on a piece of paper. They win if $\sigma_A=\sigma_B$, i.e.\ if their choices agree. Intuitively, some of the internal physical structure of quantum systems described by $[\rho]$ must be used to overcome the symmetry and to pick an element.

One convenient choice is to take the projected state $\Pi_{\rm inv}(\rho)=\Pi_{\rm inv}(\sigma)$ from Eq.~\ref{eqPiInv} as the canonical representative. In fact, in the 	quantum information context~\cite{Bartlett}, it is sometimes argued that this is the ``correct'' choice, representing an agent's state of knowledge about the quantum system if constrained by the symmetries in $\mathcal{U}_{\rm sym}$. Concretely, suppose that Alice and Bob hold physical reference frames that both break the symmetry, but these reference frames are not aligned (i.e.\ uncorrelated). Then, if Alice prepares the quantum system in state $\rho$ and sends it to Bob, Bob will assign the state $\rho':=\Pi_{\rm inv}(\rho)$ to the system he obtains. This is because his description is supposed to convey a very specific meaning: \emph{the relation of the quantum system to the concrete external reference frame in his laboratory}. But since he has no idea about this relation (even if Alice sends him a classical description of $\rho$), he must assign the mixed state $\rho'$.

However, in our context, the state description is \emph{not} meant to convey the relation of the given quantum system to a concrete external reference frame --- it is simply meant to convey a natural and useful description of the quantum system on which different observers without shared external frame may agree. Thus, in our context, it is meaningful to choose another representative $\sigma\neq \Pi_{\rm inv}(\rho)$, and mathematical or physical convenience may be a reason to do so. For example, we may choose a representative that somehow contains a smaller amount of inconvenient superpositions to arrive at a description that resembles more closely classical physics; this strategy is arguably at the heart of recent quantum formulations of the equivalence principle~\cite{GiacominiBrukner,Hardy1,Hardy2}.

Let us define a class of states that admits a particularly natural kind of representation ``relative to the $i$th particle''. To phrase the definition, we will use the notation $\overline{i}=\{1,2,\ldots,N\}\setminus\{i\}$, and thus $\ch^{\otimes N}=\ch_i\otimes \ch_{\overline i}$, where $\ch_i\simeq \ch$ and $\ch_{\overline i}\simeq \ch^{\otimes(N-1)}$.
\begin{definition}[Alignable states and observables]
Let $i\in\{1,\ldots,N\}$. A state $\rho\in\mathcal{S}(\ch^{\otimes N})$ is called ``$i$-alignable'' if there exists some $\sigma_{\overline i}\in\mathcal{S}(\ch_{\overline i})$ such that $\rho\simeq |e\rangle\langle e|_i\otimes\sigma_{\overline i}$. An analogous definition applies to observables.
\end{definition}
That is, $i$-alignable states are symmetry-equivalent to states in which the particle $i$ factors out and becomes ``located at the origin'', the unit element of the group. Interpreting the configuration of frame $i$ more generally as its orientation, such states can clearly also be aligned to ``$i$ being in orientation $g\in\cg$'', i.e.\ are symmetry-equivalent to $\ket{g}\!\bra{g}_i\otimes\tilde\sigma_{\overline i}$, where $\tilde\sigma_{\overline i}=U_g^{\otimes(N-1)}\sigma_{\overline i}U_{g^{-1}}^{\otimes(N-1)}$.  We can then regard $\sigma_{\overline i}$ as the description of the state ``relative to particle $i$ sitting in the origin'' since it is uniquely determined by $i$:
\begin{lemma}
\label{LemAlignUniqueness}
If $\rho\simeq|e\rangle\langle e|_i\otimes\sigma_{\overline i}$ and $\rho\simeq|e\rangle\langle e|_i\otimes\tau_{\overline i}$ then $\sigma_{\overline i}=\tau_{\overline i}$.
\end{lemma}
\begin{proof}
We give the proof for the case of $i=1$; the general case follows analogously. If the condition of the lemma holds, then it follows that $|e\rangle\langle e|_i\otimes\sigma_{\overline i}\simeq |e\rangle\langle e|_i\otimes\tau_{\overline i}$. Write $\sigma_{\overline 1}=\sum_{\mathbf{h},\mathbf{j}}s_{\mathbf{h},\mathbf{j}}|\mathbf{h}\rangle\langle\mathbf{j}|$ and $\tau_{\overline 1}=\sum_{\mathbf{h},\mathbf{j}}t_{\mathbf{h},\mathbf{j}}|\mathbf{h}\rangle\langle\mathbf{j}|$, then there is some $U\in\mathcal{U}_{\rm sym}$ such that
\[
   \sum_{\mathbf{h},\mathbf{j}} s_{\mathbf{h},\mathbf{j}} U|e,\mathbf{h}\rangle\langle e,\mathbf{j}|U^\dagger=\sum_{\mathbf{h},\mathbf{j}} t_{\mathbf{h},\mathbf{j}}|e,\mathbf{h}\rangle\langle e,\mathbf{j}|.
\]
But $U|e,\mathbf{h}\rangle=|g,\mathbf{h}g\rangle$ for some $g\in\cg$ (and similarly for $\mathbf{h}$ replaced by $\mathbf{j}$). Comparing this with the right-hand side shows that $U|e,\mathbf{h}\rangle=|e,\mathbf{h}\rangle$ and $U|e,\mathbf{j}\rangle=|e,\mathbf{j}\rangle$ if $s_{\mathbf{h},\mathbf{j}}\neq 0$. Thus, we can simply drop $U$ from the left-hand side, and obtain $|e\rangle\langle e|_i\otimes\sigma_{\overline i}= |e\rangle\langle e|_i\otimes\tau_{\overline i}$.
\end{proof}
Alignability does not depend on the choice of $i$:
\begin{theorem}[Ref.~\cite{QRF1}, Theorems 18 and 24]
If $\rho\in\mathcal{S}(\ch^{\otimes N})$ is $i$-alignable for some $i$, then it is $j$-alignable for every $j\in\{1,\ldots,N\}$. We will then simply call $\rho$ \emph{alignable}. Moreover, for every $i,j\in\{1,\ldots,N\}$, there is a unique symmetry $U\in\mathcal{U}_{\rm sym}$ such that $U(|e\rangle\langle e|_i\otimes\sigma_{\overline i})U^\dagger=|e\rangle\langle e|_j\otimes\sigma_{\overline j}$ for all alignable states $\rho$. If $i\neq j$, then this $U$ is a proper relation-conditional translation, i.e.\ cannot be written as an unconditional global translation. Every such $U$ induces an (up to global phase) unique unitary (``QRF transformation'') $V_{i\to j}$ such that  $V_{i\to j}\sigma_{\overline i}V_{i\to j}^\dagger=\sigma_{\overline j}$. It is given by
\[
   V_{i\to j}=\mathbb{F}_{i,j} \sum_{g\in\cg} |g^{-1}\rangle \langle g|_j\otimes U_{g^{-1}}^{\otimes(N-2)},
\]
where $\mathbb{F}$ swaps particles $i$ and $j$.
\end{theorem}
This is a special case of the QRF transformation given in Ref.~\cite{Hamette}, which generalizes the QRF transformations of Ref.~\cite{Giacomini}.

Not only can we express alignable states relative to one of the particles, but also, for example, relative to some ``center of mass'':
\begin{example}[Ref.~\cite{QRF1}, Example 20]
Recall the ``discrete circle'' of Figure~\ref{fig_cyclic} and Example~\ref{ExCyclic}, i.e.\ the cyclic group $\cg=\mathbb{Z}_n$. Let $m_1,\ldots,m_N$ be non-negative real numbers and $m:=m_1+\ldots+m_N>0$. For $\mathbf{h}\in\cg^{N-1}$, define
\[
   g(\mathbf{h}):=-\left\lfloor \frac 1 m (m_2 h_1+\ldots+m_N h_{N-1})\right\rfloor,
\]
and set $U:=\bigoplus_{\mathbf{h}}U_{g(\mathbf{h})}^{\otimes N}$. If we interpret the $m_i$ as the \emph{masses} of the particles, then this symmetry $U$ describes a change of quantum coordinates from particle $1$ to the ``center of mass''.
\end{example}
We also know the following:
\begin{lemma}
Given two alignable states $\rho\simeq|e\rangle\langle e|_i\otimes\sigma_{\overline i}$ and $\rho'\simeq |e\rangle\langle e|_i\otimes\sigma'_{\overline i}$, the following statements are equivalent:
\begin{itemize}
	\item[(i)] $\rho$ and $\rho'$ are symmetry-equivalent, i.e.\ $\rho\simeq\rho'$,
   \item[(ii)] $\rho$ and $\rho'$ are observationally equivalent, i.e.\ $\rho\sim\rho'$,
	\item[(iii)] for some (and thus every) $i\in\{1,\ldots,N\}$, the states ``relative to particle $i$'' agree, i.e.\ $\sigma_{\overline i}=\sigma'_{\overline i}$.
\end{itemize}
\end{lemma}
\begin{proof}
Equivalence of (i) and (iii) follows from Lemma~\ref{LemAlignUniqueness}. Clearly (i) implies (ii). To see that (ii) implies (iii), compute (via~\cite[Theorem 12]{QRF1}) the projection of $\rho\simeq |e\rangle\langle e|_1\otimes\sigma_{\overline 1}$ into the invariant subalgebra
\[
\Pi_{\rm inv}\left(|e\rangle\langle e|_{1}\otimes\sigma_{\overline 1}\right)=\sum_{\mathbf{h},\mathbf{j}} \frac{s_{\mathbf{h},\mathbf{j}}}{|\cg|}|\mathbf{h};\mathbf{1}\rangle\langle\mathbf{j};\mathbf{1}|+\sum_{\mathbf{h}} \frac{s_{\mathbf{h},\mathbf{h}}}{|\cg|}\Pi_{\mathbf{h};\chi\neq\mathbf{1}}
\]
using the notation of the proof of Lemma~\ref{LemAlignUniqueness} and $\Pi_{\mathbf{h};\chi\neq\mathbf{1}}:=\sum_{\chi\neq\mathbf{1}} |\mathbf{h};\chi\rangle\langle\mathbf{h};\chi|$. If $\rho\sim\rho'$ then $\Pi_{\rm inv}(\rho')$ gives the same result. But $\sigma_{\overline 1}$ can clearly be read off from the result of this projection, hence it must be the same for $\rho$ and for $\rho'$.
\end{proof}
We can characterize the alignable states as follows. This generalizes Lemma 21 of Ref.~\cite{QRF1} to mixed states.
\begin{lemma}
\label{LemAlignable}
A state $\rho\in\mathcal{S}(\ch^{\otimes N})$ is alignable if and only if for every choice of pairwise relations $\mathbf{h}$, there is \emph{at most one} classical configuration with these relations that has non-zero overlap with $\rho$.

That is, for every $\mathbf{h}$, there exists at most one $g=g(\mathbf{h})\in\cg$ such that $\langle g,\mathbf{h}g|\rho|g,\mathbf{h}g\rangle\neq 0$.
\end{lemma}
\begin{proof}
Suppose that $\rho$ is alignable, then it is in particular $1$-alignable. Thus, there is a state $\sigma_{\bar i}=\sum_{\mathbf{h},\mathbf{j}}s_{\mathbf{h},\mathbf{j}}|\mathbf{h}\rangle\langle\mathbf{j}|$ and a symmetry $U\in\mathcal{U}_{\rm sym}$ such that $\rho=U(|e\rangle\langle e|_1\otimes\sigma_{\overline 1})U^\dagger$. But $U=\bigoplus_{\mathbf{h}} U_{g(\mathbf{h})}^{\otimes N}$, hence
\[
   \rho=\sum_{\mathbf{h},\mathbf{j}} s_{\mathbf{h},\mathbf{j}} |g(\mathbf{h}),\mathbf{h}g(\mathbf{h})\rangle\langle g(\mathbf{j}),\mathbf{j} g(\mathbf{j})|
\]
which shows that $\rho$ has the claimed property.

Conversely, suppose that for every $\mathbf{h}$, there is at most one $g=g(\mathbf{h})$ such that $\langle g,\mathbf{h}g|\rho|g,\mathbf{h}g\rangle\neq 0$. For those $\mathbf{h}$ for which there is \emph{none}, choose $g(\mathbf{h})$ arbitrarily, and define the symmetry $U:=\bigoplus_{\mathbf{h}}U_{g(\mathbf{h})}^{\otimes N}$. A priori, every $\rho$ can be written
\[
   \rho=\sum_{g,\mathbf{h},l,\mathbf{j}}r_{(g,\mathbf{h}),(l,\mathbf{j})} |g,\mathbf{h}g\rangle\langle l,\mathbf{j}l|,
\]
but our special condition on $\rho$ tells us that, for fixed $\mathbf{h}$, all but one diagonal element $r_{(g,\mathbf{h}),(g,\mathbf{h})}$ must vanish. Hence, due to $\rho\geq 0$, the corresponding coherences must vanish as well: if $r_{(g',\mathbf{h}),(g',\mathbf{h})}=0$ then $r_{(g',\mathbf{h}),(l,\mathbf{j})}=0$ for all $l,\mathbf{j}$. This shows that $\rho$ has the form
\[
   \rho=\sum_{\mathbf{h},\mathbf{j}} r_{\mathbf{h},\mathbf{j}} |g(\mathbf{h}),\mathbf{h}g(\mathbf{h})\rangle\langle g(\mathbf{j}),\mathbf{j}g(\mathbf{j})|,
\]
where $r_{\mathbf{h},\mathbf{j}}:=r_{(g(\mathbf{h}),\mathbf{h}),(g(\mathbf{j}),\mathbf{j})}$. Hence $U^\dagger \rho U$ is of the form $|e\rangle\langle e|_1\otimes \sigma_{\overline 1}$.
\end{proof}
This has an interesting consequence: while $\cu_{\rm sym}$ clearly preserves the set of alignable states, the latter does not admit a linear structure.
\begin{corollary}\label{cor_nonlinalign}
The set of alignable pure states does \emph{not} constitute a Hilbert space and the set of all alignable states is not convex. Hence, the set of alignable states does not comprise the set of density matrices over some Hilbert space. 
\end{corollary}
\begin{proof}
Suppose $\ket{\psi}=\sum_{\mathbf{h}\in\cg^{N-1}}\alpha_\mathbf{h}\ket{g_\mathbf{h},\mathbf{h}g_\mathbf{h}}$ and $\ket{\psi'}=\sum_{\mathbf{h}\in\cg^{N-1}}\beta_\mathbf{h}\ket{g'_\mathbf{h},\mathbf{h}g'_\mathbf{h}}$ are two alignable states with $g_\mathbf{h}\neq g'_\mathbf{h}$ for at least one $\mathbf{h}\in\cg^{N-1}$ for which both $\alpha_\mathbf{h},\beta_\mathbf{h}\neq0$. It is clear that $a\ket{\psi}+b\ket{\psi'}$ violates the condition of Lemma~\ref{LemAlignable} for $a,b\neq0$. Similarly, one shows that convex combinations of alignable states are not in general alignable.
\end{proof}

By contrast, the set of \emph{aligned} states $\ket{e}_i\otimes\ket{\psi}_{\overline i}$ certainly generates the Hilbert space $\ket{e}_i\otimes\ch_{\overline i}$, which, however, is not invariant under $\cu_{\rm sym}$. It is instead invariant under unitaries of the form $\mathbf{1}_{i}\otimes U_{\overline i}$, where $U_{\overline i}$ is any unitary on $\ch_{\overline i}$, but these are not symmetry transformations. This fact will be at the heart of why alignable states do not admit any symmetry-preserving dynamics which can lead to non-trivial transition amplitudes between distinct interparticle relations.

\subsection{The relational trace}
In Ref.~\cite{QRF1}, we have introduced a replacement of the partial trace for $\cg$-systems: the relational trace. Starting point is an enigma that has been termed the ``paradox of the third particle''~\cite{Angelo}. Consider a $\cg$-system for the cyclic group $\cg=\mathbb{Z}_n$ where $n$ is large, and $N=2$ particles in the state
\[
   |\psi\rangle=\frac 1 {\sqrt{2}}\left(\strut |-a\rangle_1|b\rangle_2 + e^{i\theta} |a\rangle_1|-b\rangle_2\right),
\]
where $0<a,b\ll n$, and $\theta\in\mathbb{R}$. According to Lemma~\ref{LemAlignable}, this state is alignable, and so is the $3$-particle state
\[
   |\Psi\rangle=|\psi\rangle\otimes |c\rangle_3,
\]
where $0\leq c \ll n$. Now consider the following question: \emph{Is the phase $\theta$ relevant for an observer constrained by the quantum symmetries, if that observer has only access to particles $1$ and $2$?} At first sight, it seems as if observers without access to particle $3$ should hold a state that is given by the partial trace over the third particle, which is in this case $|\psi\rangle\langle\psi|$. Observers constrained by the symmetries of the $\cg$-system can physically only measure the projection of this state into the invariant subalgebra, $\Pi_{\rm inv}(|\psi\rangle\langle\psi|)$. Computing this directly turns out to give us an expression that depends on $\theta$ (\cite[Eq.\ (13)]{QRF1}). Hence the answer seems to be \emph{yes}: the phase $\theta$ \emph{is} relevant for such observers.

On the other hand, such observers cannot distinguish states $|\Psi\rangle$ from states $|\Psi'\rangle=U|\Psi\rangle$ that are related by a quantum symmetry $U\in\mathcal{U}_{\rm sym}$. Since $|\Psi\rangle$ is alignable, we can apply the symmetry map that transforms into the reference frame of particle $1$, and obtain
\[
   |\Psi'\rangle=\frac{|0\rangle_1}{\sqrt{2}}\left(\strut
      |a+b\rangle_2|a+c\rangle_3+e^{i\theta} |-a-b\rangle_2|-a+c\rangle_3
   \right).
\]
We should equally well be able to take the partial trace over the third particle in this representation. However, since this form describes particles $2$ and $3$ as maximally entangled, this leads to a mixed state, and the phase $\theta$ disappears. Paradoxically, it now seems as if the answer was \emph{no}: the phase $\theta$ is not relevant for such observers that have only access to particles $1$ and $2$.

How can this apparent paradox be resolved? The answer is that the usual partial trace is inappropriate to describe reduced states for such observers. To see how to replace it, recall why we usually apply the partial trace in the first place. In the standard case of, say, three qubits, observables $X_{12}$ of two qubits are standardly embedded in the algebra of three-qubit observables via
\[
   \Phi(X_{12}):=X_{12}\otimes\mathbf{1}.
\]
This map preserves all the structure of the observables: it is a unital $*$-homomorphism, i.e. $\Phi(\mathbf{1})=\mathbf{1}$, $\Phi(X^\dagger)=\Phi(X)^\dagger$, and $\Phi(XY)=\Phi(X)\Phi(Y)$; henceforth, we shall call a map with these properties a \emph{unital embedding}. If we now have a state $\rho_{123}$ of the three particles, then its reduction $\rho_{12}$ to the first two particles should give us the corresponding two-particle expectation values,
\[
   {\rm tr}(\rho_{12}X_{12})\stackrel ! = {\rm tr}(\rho_{123}\Phi(X_{12})) \quad\mbox{for all }X_{12},
\]
from which it follows that $\rho_{12}={\rm Tr}_3\rho_{123}$. In other words, the partial trace is the Hilbert-Schmidt adjoint $\Phi^\dagger$ of the embedding $\Phi$.

In the case of $\cg$-systems, not all observables, but only the invariant or relational ones are physically relevant (depending on whether one is interested in describing purifications of states, see Lemma~\ref{LemPurification}). To find the analog of the partial trace, we have to find the correct map $\Phi$ that embeds two into three particles (or, more generally, $N$ into $N+M$ particles). Since joining two particle groups introduces additional invariant or relational observables that cannot be obtained from those of the individual groups alone (there are non-trivial intergroup relations), this is the finite-dimensional analog of gluing two subregions in gauge theories and inquiring about how to embed subregion gauge-invariant observables into the algebra of gauge-invariant observables of the glued region (which contains more information than the subregion observables) \cite{Casini:2013rba,Donnelly:2016auv,Geiller:2019bti,Gomes:2018shn,Riello:2021lfl,Wieland:2017zkf,Wieland:2017cmf,Freidel:2020xyx,CH1}.

Since the paradox of the third particle involves alignable states, we are in particular interested in those invariant operators which arise as the \emph{invariant parts of alignable observables}, i.e.\ $\Pi_{\rm inv}(|e\rangle\langle e|_1\otimes X_{\overline 1})$. In fact, as shown in~\cite[Lemma 27]{QRF1}, we can define a natural unital embedding $\Phi$ with the following prescription:

\emph{Map the invariant part of $|e\rangle\langle e|_1\otimes X_{\overline 1}$ to the invariant part of $|e\rangle\langle e|_1\otimes X_{\overline 1}\otimes\mathbf{1}^{(M)}$.}

It is not obvious, but can be shown that this defines a valid unital embedding -- not of the full invariant $N$-particle subalgebra $\ca_{\rm inv}^{(N)}$, but of the subalgebra $\ca_{\rm alg}^{(N)}\subset\ca_{\rm inv}^{(N)}$ that is generated by the operators $\Pi_{\rm inv}(|e\rangle\langle e|_1\otimes X_{\overline 1})$ (``alg'' stands for ``invariant parts of \textbf{al}i\textbf{g}nable states''~\footnote{Note that the operators $\Pi_{\rm inv}(|e\rangle\langle e|_i\otimes X_{\bar i})$ generate the same subalgebra $\ca_{\rm alg}^{(N)}$ for every $i$, i.e.\ the definition of $\ca_{\rm alg}^{(N)}$ is independent of the choice of particle which is used as a reference. This is because $\Pi_{\rm inv}(U\bullet U^\dagger)=\Pi_{\rm inv}(\bullet)$ for all $U\in\mathcal{U}_{\rm sym}$.}). The resulting embedding map $\Phi$ respects the quantum symmetries $U\in\mathcal{U}_{\rm sym}$: starting with $U|e\rangle\langle e|_1\otimes X_{\overline 1}U^\dagger$ instead of $|e\rangle\langle e|_1\otimes X_{\overline 1}$ yields the same result, since both observables have the same invariant part. However, in contrast to the \emph{result of applying} $\Phi$, the very \emph{definition} of $\Phi$ is implicitly constructed \emph{relative to the first particle}. This can be seen as follows. As shown in Ref.~\cite[Lemma 28]{QRF1}, for every $U\in\mathcal{U}_{\rm sym}$, the following prescription defines a valid unital embedding of $\ca_{\rm alg}^{(N)}$ into $\ca_{\rm alg}^{(N+M)}$:

\emph{Map the invariant part of $U|e\rangle\langle e|_1\otimes X_{\overline 1}U^\dagger$ to the invariant part of $U|e\rangle\langle e|_1\otimes X_{\overline 1}U^\dagger\otimes\mathbf{1}^{(M)}$.}

For different $U$, we obtain unital embeddings that are in general inequivalent. In other words, \emph{the choice of QRF matters when we take the tensor product}. In Ref.~\cite{QRF1}, we give a thorough physical analysis of this inequivalence: different choices of embeddings correspond to different operational prescriptions for how to access the first $N$ of the $N+M$ particles.

Can we find an embedding whose definition is manifestly independent of any choice of QRF? The answer turns out to be \emph{yes, as long as we restrict ourselves to relational observables, i.e.\ to $\ca_{\rm phys}$}:

\begin{lemma}[Ref.~\cite{QRF1}, Lemma 31]
The map $\Phi:\ca_{\rm phys}^{(N)}\to \ca_{\rm phys}^{(N+M)}$, defined as
\[
   \Phi(X):=\Pi_{\rm phys}^{(N+M)}\left(X\otimes\mathbf{1}^{(M)}\right)\Pi_{\rm phys}^{(N+M)}
\]
is an embedding, but it is not unital. 
\end{lemma}
As shown in Ref.~\cite{QRF1}, this map can also be written $\Phi(X)=X\otimes\Pi_{\rm phys}^{(M)}$, from which it becomes obvious that it is multiplicative and preserves the adjoint. There, it is also shown that the analogous constructions for $\ca_{\rm inv}^{(N)}$ and $\ca_{\rm alg}^{(N)}$ do \emph{not} yield valid embeddings.

From this embedding, we obtain the corresponding generalization of the partial trace:
\begin{lemma}[Relational trace (Ref.~\cite{QRF1}, Eq.~(29))]
Define the \emph{relational trace} ${\rm Trel}_{(M)}:\mathcal{L}(\mathcal{H}^{\otimes (M+N)})\to\ca^{(N)}_{\rm phys}$ as
\[
   {\rm Trel}_{(M)}:=\hat\Pi^{(N)}_{\rm phys}\circ {\rm Tr}_{(M)}\circ\hat\Pi^{(N+M)}_{\rm phys},
\]
where ${\rm Tr}_{(M)}$ is the usual partial trace, and $\hat\Pi_{\rm phys}^{(N)}(\sigma):=\Pi_{\rm phys}^{(N)}\sigma \Pi_{\rm phys}^{(N)}$. Then, for every state $\rho^{(N+M)}\in\mathcal{L}(\mathcal{H}^{\otimes (M+N)})$, setting $\rho^{(N)}:={\rm Trel}_{(M)}\rho^{(N+M)}$ satisfies
\[
   {\rm tr}(\rho^{(N+M)}\Phi(X))={\rm tr}(\rho^{(N)}X)\mbox{ for all }X\in\ca_{\rm phys}^{(N)}.
\]
\end{lemma}
Note that the relational trace is \emph{trace non-increasing}, but not in general trace-preserving. This is due to the fact that states of $N+M$ particles that are fully relational (i.e.\ supported on $\ca_{\rm phys}^{(N+M)}$) do not in general have fully relational local reduced states; if there is non-trivial intergroup relational data in the state, the relational trace will project it out \emph{besides} tracing over the group of $M$ particles. Even though the relational trace will then yield \emph{subnormalized} states, these will still give us the correct reduced expectation values for all relational observables of the group of $N$ particles only.

When we interpret the ``paradox of the third particle'' scenario as having two particles embedded into three via the manifestly relational $\Phi$, then the relational trace yields a definite answer to the question of the beginning of this section: \emph{yes, the phase $\theta$ is relevant for such observers}, since it appears non-trivially in the expression ${\rm Trel}_3 |\Psi\rangle\langle\Psi|={\rm Trel}_3|\Psi'\rangle\langle\Psi'|$; see Ref.~\cite{QRF1} for the details.

The existence of the relational trace for $\ca_{\rm phys}$, and the non-existence of an analogous notion for $\ca_{\rm inv}$ or $\ca_{\rm alg}$, can be seen as a motivation to resort to the \emph{relational} states on $\ch_{\rm phys}$ for the description of relational quantum physics, rather than to alignable states. We will now further explore the relation between the two. In particular, in the next section, we will show that the two are kinematically equivalent.

\section{Kinematic equivalence of relational and alignable states}
\label{SecRelAl}

In Subsection~\ref{SubsecAlignable}, we have introduced a class of states for $\mathcal{G}$-systems which have a natural representation ``relative to the $i$th particle'': the alignable states. External observers who do not share a common reference frame can obtain a common description of such states by simply agreeing to describe the state relative to one of the particles. These states transform non-trivially under symmetry transformations. We have also discussed the class of quantum states that is supported on the subspace $\mathcal{H}_{\rm phys}$: the \emph{relational states}. These, by contrast, are invariant, i.e.\ $U|\psi\rangle=|\psi\rangle$ for all $U\in\mathcal{U}_{\rm sym}$ and $\ket{\psi}\in\mathcal{H}_{\rm phys}$. Thus, the representation of these states is the same relative to all external frames, and external observers do not need a shared reference frame in the first place to understand each other's descriptions of such states. The same is true for purifications of mixed states on $\ch_{\rm phys}$ as shown in Lemma~\ref{LemPurification}.

Though these two types of states thus seem \emph{a priori} very different, there is an equivalence between them. More precisely, as shown in~\cite[Lemma 22]{QRF1}, if we have an alignable state $\rho$ with
\[
   \rho\sim \sum_{\mathbf{h},\mathbf{j}\in\mathcal{G}^{N-1}} r_{\mathbf{h},\mathbf{j}} |e,\mathbf{h}\rangle\langle e,\mathbf{j}|
\]
then its projection into the physical subspace is
\begin{equation}
   \hat\Pi_{\rm phys}(\rho):=\Pi_{\rm phys}\rho\Pi_{\rm phys}=\sum_{\mathbf{h},\mathbf{j}\in\mathcal{G}^{N-1}} \frac{r_{\mathbf{h},\mathbf{j}}}{|\mathcal{G}|} |\mathbf{h};\mathbf{1}\rangle\langle\mathbf{j};\mathbf{1}|,
   \label{eqProj}
\end{equation}
and this gives us a one-to-one correspondence between the physical states and the alignable states as represented relative to any one of the particles. Furthermore, all external-frame-independent information of $\rho$ is contained in its projection $ \Pi_{\rm phys}\rho\Pi_{\rm phys}$ onto the physical subspace~\cite[Lemma 22]{QRF1}. Hence, the algebra $\ca_{\rm phys}$ of relational observables is tomographically complete for the invariant information in alignable states.

This suggests that the same set of (external-frame-independent) physical scenarios can be equivalently described via alignable and relational states, and that the choice is purely conventional. In this section, we shall demonstrate explicitly that this is indeed the case \emph{kinematically}. However, in the next section, we shall then explain why this kinematical equivalence of alignable and relational states is in general not dynamically stable, depending on what sort of symmetry-preserving dynamics the total $N$-particle system might be subjected to.

More concretely, in this section we shall:
\begin{itemize}
\item[(i)] establish the explicit reduction transformations mapping relational into alignable states and their inverses,
\item[(ii)] use these same maps to reversibly transform between relational and alignable observables, and
\item[(iii)] use these reduction maps to prove equivalence of the QRF transformations derived in Ref.~\cite{QRF1} with the `quantum coordinate transformations' of the perspective-neutral approach \cite{Vanrietvelde,Vanrietvelde2,Hoehn:2018aqt,Hoehn:2018whn,Hoehn:2019owq,Hoehn:2020epv,Hoehn:2021wet,periodic}.
\end{itemize}
This will manifest that relational states and observables are simply (a representation of) the symmetry equivalence classes of alignable states and observables. In particular, specifically aligned states can be viewed as gauge-fixed reductions of the gauge-invariant relational states, while the latter are the coherent group averages of the differently aligned states. This is ultimately the reason why we can regard the relational states as \emph{perspective-neutral states}: they encode all the internal QRF perspectives\,---\,hence are `frame perspective neutral'\,---\,and provide the link between the latter.
Moreover, since the QRF transformations constructed in Refs.~\cite{Giacomini,Hamette} are equivalent to the ones derived in Ref.~\cite{QRF1} (for finite Abelian groups), result (iii) will in turn also establish equivalence of the former with the quantum coordinate transformations of the perspective-neutral approach, thereby linking three \emph{a priori} different formulations of QRF transformations.\footnote{Equivalence between the quantum coordinate transformations of the perspective-neutral approach and the QRF transformations of Ref.~\cite{Giacomini} has previously been shown for the continuum translation group in Refs.~\cite{Vanrietvelde,Vanrietvelde2}, while equivalence between the QRF transformations of Refs.~\cite{Giacomini} and \cite{Hamette} for the same group was demonstrated in Ref.~\cite{Hamette}. Here, we reveal the explicit equivalence between the transformations in the perspective-neutral approach and the ones in Ref.~\cite{Hamette} for finite groups. Equivalence of the QRF transformations for general groups of Ref.~\cite{Hamette} with those of the perspective-neutral approach will be demonstrated in Ref.~\cite{all}.}

In order to establish these various equivalences, we shall adapt the method of Refs.~\cite{Hoehn:2019owq,Hoehn:2020epv,periodic} to finite Abelian groups. To this end, we slightly generalize our previous description (and that in Ref.~\cite{Hamette}): rather than aligning or describing states and observables relative to the reference frame particle only being ``in the origin $e$", we will include descriptions relative to the frame particle being ``in orientation $g\in\cg$". This will permit us to formulate two unitarily equivalent reductions from relational to alignable states: a Schr\"odinger picture (or Page-Wootters) reduction and a Heisenberg picture (or quantum symmetry) reduction.

\subsection{Relational observables}

Suppose we choose subsystem $i$ as the reference system. We would like to construct a relational observable that, loosely speaking, encodes the question `what is the value of some observable $f_{\overline{i}}\in\mathcal{L}(\ch_{\overline{i}})$ on the composite system $\overline{i}$ of interest when the reference system $i$ is in orientation $g$?'. In particular, this relational observable shall be contained in $\ca_{\rm phys}$.\footnote{In continuous systems, relational Dirac observables are typically defined also ``off-shell'' of the constraint surface, i.e.\ as elements of $\ca_{\rm kin}$ rather than $\ca_{\rm phys}$, which means as incoherently rather than coherently group-averaged observables in contrast to here, e.g.\ see Refs.~\cite{Hoehn:2019owq,Hoehn:2020epv,periodic,Hoehn:2021wet,all,Chataignier,Chataignier2,Chataignier3} within the context of constraint quantization and Refs.~\cite{Bartlett,Loveridge2017,Loveridge2018,Miyadera} within the context of quantum information theory and quantum foundations. In constraint quantization, one is usually interested in their actions on $\ch_{\rm phys}$, and since the two group-averaging procedures agree on this space \cite{QRF1}, the choice of definition then does not make a difference. Here, for notational simplicity, we restrict their definition \emph{ab initio} to $\ca_{\rm phys}$ as some of their algebraic properties only hold when acting on $\ch_{\rm phys}$. This permits us to write all their algebraic relations without additional restrictions.}  We define it as the projection of aligned observables into $\ca_{\rm phys}$ (cf.\ \cite{Hoehn:2019owq,Hoehn:2020epv,periodic}):
\begin{definition}
The \emph{relativization map} relative to reference system $i\in\{1,\ldots,N\}$ being in orientation $g$ is
\begin{eqnarray}
F_{\bullet}^{(i)}(g):\mathcal{L}(\ch_{\overline{i}})&\rightarrow&\ca_{\rm phys}\nn\\
f_{\overline{i}}&\mapsto& F_{f_{\overline{i}}}^{(i)}(g)
:=|\mathcal{G}|\cdot \hat\Pi_{\rm phys}(|g\rangle\langle g|_i \otimes f_{\overline i}).\quad\quad\strut\label{homo}
\end{eqnarray}
\end{definition}
\noindent The normalization factor $|\cg|$ turns out to be necessary for the following lemma to hold. It follows directly from Eq.~(\ref{eqProj}) and Eq.~(\ref{homo}) (see also~\cite[Theorem 25]{QRF1}):
\begin{lemma}
The relationalization maps $F_{\bullet}^{(i)}(g)$ are algebra homomorphisms:
\[
F^{(i)}_{a_{\overline{i}}+b_{\overline{i}}\cdot c_{\overline{i}}}(g)=F^{(i)}_{a_{\overline{i}}}(g)+F^{(i)}_{b_{\overline{i}}}(g)\cdot F^{(i)}_{c_{\overline{i}}}(g).
\]
\end{lemma}
Relational observables thus preserve the algebraic structure of the observables of the subsystem $\overline{i}$ of interest.

\subsection{Reducing to a frame perspective the Page-Wootters way}
\label{SubsecPageWootters}

We would now like to construct a `quantum coordinate map' from the physical Hilbert space $\ch_{\rm phys}$ to the Hilbert space $\ch_{\overline i}$, i.e.\ a map that gives us the description of the remaining systems $\overline i$ relative to system $i$. We will describe two unitarily equivalent such maps, beginning here with a reduction into a \emph{relational Schr\"odinger picture} relative to $i$, adapting the method of Refs.~\cite{Hoehn:2019owq,Hoehn:2020epv,periodic} to the finite Abelian group case. This can be achieved by conditioning  the relational states, which we henceforth call perspective-neutral states as $i$ is arbitrary, onto the `gauge-fixing condition' $g_i=g$ in extension of the Page-Wootters formalism for quantum clocks \cite{Page}:
\begin{definition}\label{def_Sred}
We define the \emph{Schr\"odinger reduction map} $\mathcal{R}_{\mathbf{S},i}(g):\mathcal{H}_{\rm phys}\to\mathcal{H}_{\overline i}$ from perspective-neutral states to the description of $\overline{i}$ relative to frame $i\in\{1,\ldots,N\}$  in orientation $g$ via
\ba
\calr_{\mathbf{S},i}(g):=\sqrt{|\cg|}\bra{g}_i\otimes \mathbf{1}_{\overline{i}}.\label{redmap}
\ea
As a simple example and for later reference, we set $\calr:=\calr_{\mathbf{S},1}(e)$.
\end{definition}
To get a glimpse on what these maps do, note that $\calr|\mathbf{h};\mathbf{1}\rangle=|\mathbf{h}\rangle$: that is, $\calr$ maps the relational state where systems $\overline 1$ hold relation $\mathbf{h}$ with system $1$ to a state of $\overline 1$ that literally \emph{is} in configuration $\mathbf{h}$ (implicitly, thus, moving system $1$ into the origin $e$). For alignable states, the map $\calr_{\mathbf{S},i}(g)$ ``reverses'' the projection onto the physical subspace and translates the result by the group element $g$. To see this, consider some (normalized) alignable state $|\psi\rangle\simeq |e\rangle_i\otimes|\varphi\rangle_{\overline i}$, and define $|\psi\rangle_{\rm phys}:=\sqrt{|\mathcal{G}|}\Pi_{\rm phys}|\psi\rangle$, where the prefactor is chosen such that the resulting state is normalized. Using~\cite[Lemma 11]{QRF1}, we obtain
\begin{eqnarray}
\mathcal{R}_{\mathbf{S},i}(g)|\psi\rangle_{\rm phys}&=& \frac 1 {\sqrt{|\mathcal{G}|}}\mathcal{R}_{\mathbf{S},i} (g)\sum_{g'\in\mathcal{G}} U_{g'}^{\otimes N}|e\rangle_i\otimes|\varphi\rangle_{\overline i}\nn\\
&=& \langle g|_i\otimes\mathbf{1}_{\overline i}\sum_{g'\in\mathcal{G}} |g'\rangle_i\otimes U_{g'}^{\otimes(N-1)} |\varphi\rangle_{\overline i}\nn\\
&=&|\varphi(g)\rangle_{\overline i},\nn
\end{eqnarray}
where 
\[
|\varphi(g)\rangle_{\overline i}:=U_{g}^{\otimes(N-1)} |\varphi\rangle_{\overline i}.
\]
This conditional state satisfies the covariance property $\ket{\varphi(g)}_{\overline{i}}=U_{gg'^{-1}}^{\otimes(N-1)}\ket{\varphi(g')}_{\overline{i}}$ in analogy to the Schr\"odinger evolution of the Page-Wootters formalism \cite{Page}, which is why we call the reduced formulation the relational Schr\"odinger picture.

It will be convenient to consider the inverse of the Schr\"odinger reduction map:
\begin{lemma}\label{lem_invred}
Every Schr\"odinger reduction map $\mathcal{R}_{\mathbf{S},i}(g)$ is unitary, and the inverse $\mathcal{R}_{\mathbf{S},i}^{-1}(g):\mathcal{H}_{\overline i}\to\mathcal{H}_{\rm phys}$ can be written
\begin{equation}
\calr_{\mathbf{S},i}^{-1}(g)=\sqrt{|\cg|}\,\Pi_{\rm phys}\left(\ket{g}_i\otimes \mathbf{1}_{\overline{i}}\right).\label{invredmap}	
\end{equation}
\end{lemma}
\begin{proof}
To prove $\mathcal{R}_{\mathbf{S},i}(g)\mathcal{R}_{\mathbf{S},i}^{-1}(g)=\mathbf{1}_{\overline i}$, replace $\Pi_{\rm phys}$ by the right-hand side of Eq.~(\ref{eqPiPhys}). Since $\dim \mathcal{H}_{\overline i}=\dim \mathcal{H}_{\rm phys}$, we must also have $\mathcal{R}_{\mathbf{S},i}^{-1}(g)\mathcal{R}_{\mathbf{S},i}(g)=\mathbf{1}_{\rm phys}$. Note that $\calr_{\mathbf{S},i}(g)^\dagger \calr_{\mathbf{S},i}(g) = |\cg|\,|g\rangle\langle g|_i\otimes\mathbf{1}_{\overline i}$ as operators on $\ch^{\otimes N}$. Using the twirling identity again, it follows that $\Pi_{\rm phys}\calr_{\mathbf{S},i}(g)^\dagger \calr_{\mathbf{S},i}(g)\Pi_{\rm phys}=\Pi_{\rm phys}$, which is equivalent to $\langle \calr_{\mathbf{S},i}(g)\varphi,\calr_{\mathbf{S},i}(g)\psi\rangle=\langle\varphi,\psi\rangle$ for all $|\varphi\rangle,|\psi\rangle\in\ch_{\rm phys}$, the definition of unitarity.
\end{proof}
The reduction maps allow us to rewrite the relational observables in a convenient manner. The following two lemmas are adaptations of results in Refs.~\cite{Hoehn:2019owq,Hoehn:2020epv,periodic}.
\begin{lemma}\label{lem_obs}
The conjugation of a relational observable with the Schr\"odinger reduction map relative to frame $i$ yields the corresponding $\overline{i}$-observable, i.e.
\ba
\calr_{\mathbf{S},i}(g)\,F_{f_{\overline{i}}}^{(i)}(g)\,\calr_{\mathbf{S},i}^{-1}(g)=f_{\overline{i}}.\nn
\ea
Since the Schr\"odinger reduction map is unitary, this defines a one-to-one correspondence between the relational observables and the $\overline i$-observables. 
\end{lemma}
\begin{proof}
Using Eq.~\eqref{eqPiPhys} for $\Pi_{\rm phys}$, the statement follows from
\ba
\calr_{\mathbf{S},i}(g)\,F_{f_{\overline{i}}}^{(i)}(g)\,\calr_{\mathbf{S},i}^{-1}(g)&=&|\cg|\,(\bra{g}_i\otimes {\mathbf{1}}_{\overline{i}})\sum_{g'\in\cg} U_{g'}^{\otimes N}\nn\\&&\times(\ket{g}\!\bra{g}_i\otimes f_{\overline{i}})
\, \Pi_{\rm phys}(\ket{g}_i\otimes \mathbf{1}_{\overline{i}})\nn \\
&=& |\cg|(\langle g|_i\otimes f_{\overline i})\Pi_{\rm phys}(|g\rangle_i\otimes\mathbf{1}_{\overline i})\nn\\
&=&f_{\overline i}.\nn
\ea
\vskip -1.6em
\end{proof}
This translates into an equivalence of expectation values:
\begin{lemma}\label{lem_expec}
Let $\ket{\psi(g)}_{\overline{i}}=\calr_{\mathbf{S},i}(g)\ket{\psi}_{\rm phys}$ be the state of $\overline{i}$ when reference system $i$ is in orientation $g$ and similarly for $\ket{\phi(g)}_{\overline{i}}$.
Then $\braket{\psi_{\rm phys}|\,F_{f_{\overline{i}}}^{(i)}(g)\,|\phi_{\rm phys}}=\braket{\psi(g)_{\overline{i}}|\,f_{\overline{i}}\,|\phi(g)_{\overline{i}}}$.
\end{lemma}
\begin{proof}
Using definition~(\ref{homo}), we get
\ba
&&\braket{\psi_{\rm phys}|\,F_{f_{\overline{i}}}^{(i)}(g)\,|\phi_{\rm phys}}\nn\\
&&\q\q\q=\langle\calr_{\mathbf{S},i}(g)^{-1}\psi(g)_{\overline i}|F_{f_{\overline i}}^{(i)}(g)|\calr_{\mathbf{S},i}(g)^{-1}\phi(g)_{\overline i}\rangle\nn,
\ea
and then the result follows from the unitarity of $\calr_{\mathbf{S},i}(g)$ and Lemma~\ref{lem_obs}.
\end{proof}
In other words, the expectation values of the relational observables on the perspective-neutral Hilbert space $\ch_{\rm phys}$ are equivalent to the expectation values of the $\overline{i}$-observables relative to system $i$. 

It is now clear how to use the Schr\"odinger reduction maps to establish an isomorphism between relational and alignable states and observables. Defining $\mathcal{\overline R}_{\mathbf{S},i}(g):\ch_{\rm phys}\rightarrow\ket{g}_i\otimes\ch_{\overline i}\subset\ch$ by
\ba
\mathcal{\overline R}_{\mathbf{S},i}(g):=\ket{g}_i\otimes\calr_{\mathbf{S},i}(g)\label{modSmap}
\ea
and its inverse by
\ba
\mathcal{\overline R}^{-1}_{\mathbf{S},i}(g):=\calr_{\mathbf{S},i}^{-1}(g)\left(\bra{g}_i\otimes\mathbf{1}_{\overline i}\right)\,,\nn
\ea
we find
\ba
\mathcal{\overline{R}}_{\mathbf{S},i}(g)\,\ket{\psi}_{\rm phys}=\ket{g}_i\otimes\ket{\phi(g)}_{\overline i}\nn
\ea
and
\ba
\mathcal{\overline R}_{\mathbf{S},i}(g)\,F^{(i)}_{f_{\overline i}}(g)\,\mathcal{\overline R}^{-1}_{\mathbf{S},i}(g) = \ket{g}\!\bra{g}_i\otimes f_{\overline i}\,.\nn
\ea
For later reference, we set $\mathcal{\overline R}_i:=\mathcal{\overline R}_{\mathbf{S},i}(e)$ and $\mathcal{\overline R}:=\mathcal{\overline R}_1$.

\subsection{QRF transformations as quantum coordinate changes}

It is clear that the QRF transformation from reference system $i$ in orientation $g_i$ to reference system $j$ in orientation $g_j$ is $V_{i\to j}(g_i,g_j):\ch_{\overline{i}}\rightarrow\ch_{\overline{j}}$, where
\ba
V_{i\to j}(g_i,g_j):=\calr_{\mathbf{S},j}(g_j)\cdot \calr_{\mathbf{S},i}^{-1}(g_i),\label{CFM}
\ea
and $\ket{\psi(g_j)}_{\overline{j}}=V_{i\to j}(g_i,g_j)\,\ket{\psi(g_i)}_{\overline{i}}$. 
This transformation thus takes the same compositional form as coordinate changes on a manifold.
In particular, the Schr\"odinger reduction maps assume the role of `quantum coordinate maps' and the transformation links the internal frame perspectives on $\ch_{\overline i}$ and $\ch_{\overline j}$ via the perspective-neutral Hilbert space $\ch_{\rm phys}$ which thus assumes the analogous role to the manifold.

An observable that takes the form $f_{\overline i}$ relative to system $i$ can be transformed into the description relative to $j$ via
\ba
   f_{\overline j}(g_i,g_j):=V_{i\to j}(g_i,g_j) f_{\overline i}V_{j\to i}(g_j,g_i).\label{obstrans}
\ea
In analogy to state transformations, this transformation maps $\mathcal{L}(\ch_{\overline i})$ \emph{via} the algebra of relational observables $\ca_{\rm phys}$ into $\mathcal{L}(\ch_{\overline j})$. This is evident from Eq.~\eqref{CFM} and Lemma~\ref{lem_obs}. The image observable $f_{\overline j}$ in $j$'s perspective will generally depend non-trivially on both frame orientations $g_i,g_j$ even if $f_{\bar{i}}$ in $i$'s perspective carries no further dependence on  $g_i$.  As argued in Refs.~\cite{Hoehn:2019owq,Hoehn:2020epv}, the $g_j$-dependence of the transformed observable in $j$-perspective can be viewed as an indirect self-reference of frame $j$.

Given that the frame change map will generally transform observables non-trivially, it is interesting to ask when observables take the \emph{same} form in $i$'s and $j$'s perspective. The following observation, which is an adaptation of \cite[Corollary 4]{Hoehn:2019owq}, shows this to be the case when the reduced observable is translation-invariant. 

\begin{lemma}
Suppose $f_{\overline i}=\mathbf{1}_j\otimes f_{\overline{ij}}$. Then $f_{\overline j}(g_i,g_j)=\mathbf{1}_i\otimes f_{\overline {ij}}$, i.e.\ the transformed observable of $\overline{ij}$\,---\,the complement of $ij$ with $i\neq j$\,---\,appears in the same form relative to both $i$ and $j$ (and is in particular $g_i,g_j$-independent) if and only if $[ f_{\overline{ij}},U_g^{\otimes(N-2)}]=0$, for all $g\in\cg$. 
\end{lemma}

\begin{proof}
For $f_{\overline i}=\mathbf{1}_j\otimes f_{\overline{ij}}$, Eq.~\eqref{obstrans} yields
\ba
f_{\overline j}(g_i,g_j)&=&\sum_{g,g'\in\cg}\,\ket{gg_i}\!\bra{g_ig'^{-1}}_i\nn\\
&&\q\q\q\otimes\bra{g_j}_jU_g^{\otimes(N-1)}\, f_{\overline i}\,U_{g'}^{\otimes(N-1)}\ket{g_j}_j\nn\\
&=&\sum_{g\in\cg} U_g^{\otimes(N-1)}\left(\ket{g_i}\!\bra{g_i}_i\otimes f_{\overline{ij}}\right) U_{g^{-1}}^{\otimes(N-1)}\label{obstrans1}.
\ea
The statement that $f_{\overline j}(g_i,g_j)=\mathbf{1}_i\otimes f_{\overline{ij}}$ if $[ f_{\overline{ij}},U_g^{\otimes(N-2)}]=0$ now follows immediately, using the resolution of the identity. 
 Conversely, suppose 
 \ba 
 f_{\overline j}(g_i,g_j)=\mathbf{1}_i\otimes f_{\overline{ij}}=\sum_{g\in\cg} U_g\left(\ket{g_i}\!\bra{g_i}_i\right)U^\dag_g\otimes f_{\overline{ij}}\,.\nn 
 \ea
Comparing with Eq.~\eqref{obstrans1}, this is only possible if $[ f_{\overline{ij}},U_g^{\otimes(N-2)}]=0$, for all $g\in\cg$.
  \end{proof}

This construction of reduction maps and quantum coordinate changes for states and observables constitutes the perspective-neutral approach to quantum frame covariance \cite{Vanrietvelde,Vanrietvelde2,Hoehn:2018aqt,Hoehn:2018whn,Hoehn:2019owq,Hoehn:2020epv,Hoehn:2021wet,periodic}  for finite Abelian groups. We see how the perspective-neutral Hilbert space $\ch_{\rm phys}$ and observable algebra $\ca_{\rm phys}$ encode and link the different internal QRF perspectives. The perspective-neutral structures can be viewed as an external frame-independent description of the physics before choosing an internal QRF relative to which one wishes to describe the $N$-particle system.

The quantum coordinate changes can be written more explicitly:
\begin{lemma}\label{lem_CFM}
The change of frame map \eqref{CFM} takes the form
\[
V_{i\to j}(g_i,g_j)=\sum_{g\in\cg}\,\ket{gg_i}_i\otimes\bra{g_jg^{-1}}_j\otimes U_g^{\otimes(N-2)}\quad(i\neq j).
\]
\end{lemma}
\begin{proof}
Invoking the definitions \eqref{redmap} and \eqref{invredmap} of the Schr\"odinger reduction maps and their inverses, we have
\begin{eqnarray*}
V_{i\to j}(g_i,g_j)&=&|\cg|\left(\bra{g_j}_j\otimes \mathbf{1}_{\overline{j}}\right)\,\Pi_{\rm phys}\,\left(\ket{g_i}_i\otimes \mathbf{1}_{\overline i}\right)\\
&=& \sum_{g\in\cg} \langle g_j|_j\otimes \mathbf{1}_{\overline j}U_g^{\otimes N}|g_i\rangle_i\otimes\mathbf{1}_{\overline i}\\
&=& \sum_{g\in\cg} U_g^{\otimes(N-2)}\otimes\langle g_j|_j\otimes\mathbf{1}_i U_g^{\otimes 2}|g_i\rangle_i\otimes\mathbf{1}_j
\end{eqnarray*}
which gives the claimed identity.
\end{proof}

Specifically, setting both frame orientations to the identity, $g_i=g_j=e$, $V_{i\to j}:=V_{i\to j}(e,e)$ coincides with the QRF transformation established in Ref.~\cite[Theorems 18 \& 24]{QRF1}. Similarly,
$
\mathcal{\overline R}_j\cdot\mathcal{\overline R}_i^{-1}=\mathbb{F}_{i,j}\ket{e}\!\bra{e}_i\otimes V_{i\to j},\nn
$
where $\mathbb{F}_{i,j}$ swaps particles $i$ and $j$, coincides with the (finite Abelian case of the) QRF transformation in Ref.~\cite{Hamette}.

\subsection{Reducing to a frame perspective via symmetry reduction}
\label{SubsecReducing}
Let us summarize an alternative, but unitarily equivalent reduction method from the perspective-neutral formulation into a QRF perspective which results in a relational Heisenberg picture. This is the quantum analog of phase space symmetry reduction by gauge fixing and used in the original formulation of the perspective-neutral approach \cite{Vanrietvelde,Vanrietvelde2,Hoehn:2018aqt,Hoehn:2018whn}. It involves an additional step which shifts all non-redundant information in perspective-neutral states into the (kinematical) $\overline i$ tensor factor and renders the chosen QRF's degrees of freedom redundant. This can be achieved by a kinematical disentangler\footnote{It is important to note that this entanglement refers to the tensor product structure of the kinematical Hilbert space $\ch^{\otimes N}$ which $\ch_{\rm phys}$ does \emph{not} inherit. The entanglement we refer to here is thus only discernible relative to an external observer with access to an external frame, however, is not internally detectable with observables in $\ch_{\rm phys}$, see Refs.~\cite{Hoehn:2019owq,Hoehn:2021wet} for further discussion.} (`trivialization') that can be written as a frame-orientation-conditional shift operator \cite{Hoehn:2021wet}
\ba
\mathcal{T}_i:=\sum_{g\in\cg}\ket{g}\!\bra{g}_i\otimes U_{g^{-1}}^{\otimes(N-1)}\,,\nn
\ea
and is unitary for the present case of finite Abelian groups.
\begin{lemma}
The trivialization map transforms the coherent group averaging projector such that it only acts on the chosen frame $i\in\{1,\ldots,N\}$ 
\[
\ct_i\,\Pi_{\rm phys}\,\ct_i^\dag=\Pi_{\rm phys}^{(i)}\otimes\mathbf{1}_{\overline i}\,,
\]
where $\Pi_{\rm phys}^{(i)}:=\frac{1}{|\cg|}\sum_{g\in\cg}\,U^{(i)}_g$. This kinematically disentangles the system of interest $\overline i$ from the frame $i$ in perspective-neutral states 
\[
\ct_i\,\ket{\psi}_{\rm phys}=\frac{1}{\sqrt{|\cg|}}\sum_{g\in\cg}\ket{g}_i\otimes\ket{\psi}_{\overline i}\,,
\]
where $\otimes$ denotes the tensor product  of $\ch^{\otimes N}$ between the $i$ and $\overline i$ factors, and $\ket{\psi}_{\overline i}:=\calr_{\mathbf{S},i}(e)|\psi\rangle_{\rm phys}$ is a `relational Heisenberg state'.
\end{lemma}
The proof is straightforward and omitted. The frame orientations thus becomes coherently averaged out. 

The second step of the symmetry reduction then involves a conditioning on the frame orientation as before, which now just removes the redundant  frame information. In analogy to Definition~\ref{def_Sred} this yields:
\begin{definition}
We define the \emph{Heisenberg reduction map} $\mathcal{R}_{\mathbf{H},i}:\mathcal{H}_{\rm phys}\to\mathcal{H}_{\overline i}$ {from perspective-neutral states to the description of $\overline{i}$ relative to frame $i\in\{1,\ldots,N\}$  in orientation $g$} via
\ba
\calr_{\mathbf{H},i}:=\sqrt{|\cg|}\left(\bra{g}_i\otimes \mathbf{1}_{\overline{i}}\right)\ct_i.\label{Hredmap}
\ea
\end{definition}
Let us provide some justification for why we refer to the result as a `relational Heisenberg picture'.
First, this reduction map is independent of the frame orientation $g$ in its definition since, for finite Abelian groups, we simply have
\[
\calr_{\mathbf{H},i}=U^{\otimes(N-1)}_{g^{-1}}\calr_{\mathbf{S},i}(g)=\calr_{\mathbf{S},i}(e)\q\forall\,g\in\cg\,,
\]
where the last equality follows immediately when acting on the domain $\ch_{\rm phys}$ of the two maps. This immediately establishes the unitary  equivalence of the two reduction methods for all frame orientations $g\in\cg$.\footnote{For non-compact groups, this equivalence is not trivial~\cite{Hoehn:2019owq,Hoehn:2020epv,periodic}.} In particular, in the image of the Heisenberg reduction, observables now transform covariantly as `Heisenberg picture' operators for $\cg$, i.e.\ 
\ba
\calr_{\mathbf{H},i}\,F_{f_{\bar{i}}}^{(i)}(g)\,\calr_{\mathbf{H},i}^{-1}=f_{\bar{i}}(g)=U_{g^{-1}}^{\otimes (N-1)} f_{\bar i} \,U_{g}^{\otimes(N-1)}\,,\nn
\ea
while states remain fixed, constituting `Heisenberg states' for $\cg$:
\ba
\calr_{\mathbf{H},i}\ket{\psi}_{\rm phys}=\ket{\varphi}_{\overline i}=\ket{\varphi(e)}_{\overline i}\,.\nn
\ea
In conjunction, this section translates the ``trinity'' of equivalent descriptions of physics relative to internal QRFs established in Refs.~\cite{Hoehn:2019owq,Hoehn:2020epv,periodic}\,---\,the perspective-neutral, relational Schr\"odinger and relational Heisenberg pictures\,---\,into the finite Abelian group context. In addition, it establishes equivalence of this trinity with the fourth formulation in terms of alignable states and observables~\cite{QRF1,Hamette}.

\section{Dynamical inequivalence of relational and alignable descriptions}
\label{Section:Time}
We have established a complete formal isomorphism between the relational states and observables and those relative to some subsystem $i$. Therefore, it seems like the alignability picture and the perspective-neutral picture give us identical possibilities to describe relational quantum physics.

However, we will now see that this is not quite true if we consider time evolution. Recall our starting point: we have a quantum system (``$\mathcal{G}$-system'') with the property that external observers without access to the external relatum cannot distinguish states $\rho$ and $U\rho U^\dagger$, where $U\in\mathcal{U}_{\rm sym}$ is a quantum symmetry transformation.

Suppose our system evolves unitarily via $U(t)=\exp(-i H t)$ for some Hamiltonian $H$. If two states $\rho$ and $\sigma$ are indistinguishable as a consequence of the quantum symmetry, then the same should hold for $U(t)\rho U(t)^\dagger$ and $U(t)\sigma U(t)^\dagger$. Otherwise, if these states were distinguishable without access to the external relatum, then this would contradict our initial assumption that observers so constrained have no means to distinguish the two states.

Thus, this leads us to the question of which Hamiltonians -- or, more generally, unitaries -- preserve equivalence of states.

\subsection{Symmetry-preserving unitaries}
Recall that we have two notions of equivalence, and this yields two notions of equivalence-preserving maps.
\begin{definition}
Let $W$ be a unitary map on $\mathcal{H}^{\otimes N}$. We say that $W$ \emph{preserves symmetry equivalence}, or that $W$ is \emph{SE-preserving}, if for every pair of states
\[
   \rho\simeq \sigma \Rightarrow W\rho W^\dagger \simeq W\sigma W^\dagger.
\]
We say that $W$ \emph{preserves observational equivalence}, or that $W$ is \emph{OE-preserving}, if
\[
   \rho\sim \sigma \Rightarrow W\rho W^\dagger \sim W\sigma W^\dagger.
\]
\end{definition}
The groups of equivalence-preserving maps can be partially characterized as follows:
\begin{lemma}
\label{LemChain}
Let $W$ be a unitary map on $\mathcal{H}^{\otimes N}$. We have the following chain of implications:
\begin{eqnarray*}
W&\in&\mathcal{A}_{\rm inv}\\
&\Downarrow& \\
W\mathcal{U}_{\rm sym} W^\dagger &=&\mathcal{U}_{\rm sym}\\
& \Downarrow & \\
W&\mbox{is}&\mbox{SE-preserving} \\
& \Downarrow \not \Uparrow& \\
W^\dagger \mathcal{A}_{\rm inv} W &=&\mathcal{A}_{\rm inv}\\
&\Updownarrow&\\
W&\mbox{is}&\mbox{OE-preserving}.
\end{eqnarray*}
The symbol $\not\Rightarrow$ here indicates that there exist finite Abelian groups $\mathcal{G}$ (e.g.\ $\mathbb{Z}_4$) for which this implication does not hold.
\end{lemma}
\begin{proof}
Clearly, if $W\in\mathcal{A}_{\rm inv}$ then $[W,U]=0$ for all $U\in\mathcal{U}_{\rm sym}$, hence $W\mathcal{U}_{\rm sym} W^\dagger=\mathcal{U}_{\rm sym}$. Let $\rho\simeq\sigma$, then there exists some $U\in\mathcal{U}_{\rm sym}$ with $\sigma=U\rho U^\dagger$. Thus
\[
   W\sigma W^\dagger=(WUW^\dagger)W\rho W^\dagger (WUW^\dagger)^\dagger,
\]
and $WUW^\dagger\in\mathcal{U}_{\rm sym}$, hence $W\rho W^\dagger\simeq W\sigma W^\dagger$. This proves that $W$ is SE-preserving.

Now let $\rho\in\mathcal{A}_{\rm inv}$ be any state, then $U\rho U^\dagger=\rho$ for all $U\in\mathcal{U}_{\rm sym}$, hence $\rho$ is only symmetry-equivalent to itself. Let $\sigma$ be any state such that $W^\dagger \rho W \simeq \sigma$. Since $W$ is SE-preserving, it follows that $\rho\simeq W\sigma W^\dagger$ and thus $\sigma=W^\dagger \rho W$. Thus $W^\dagger \rho W$ is only SE-equivalent to itself, i.e.\ $U W^\dagger \rho W U^\dagger = W^\dagger \rho W$ for all $U\in\mathcal{U}_{\rm sym}$. This implies $[U,W^\dagger \rho W]=0$ for all $U\in\mathcal{U}_{\rm sym}$, i.e.\ $W^\dagger\rho W \in\mathcal{A}_{\rm inv}$. Since the states linearly span all of $\mathcal{A}_{\rm inv}$, it follows that $W^\dagger \mathcal{A}_{\rm inv} W = \mathcal{A}_{\rm inv}$. In turn, this implies
\begin{eqnarray*}
\rho\sim\sigma &\Leftrightarrow& \tr(A\rho)=\tr(A\sigma)\quad\mbox{for all }A\in\mathcal{A}_{\rm inv}\\
&\Leftrightarrow& \tr(W^\dagger A W \rho)=\tr(W^\dagger A W\sigma)\quad\mbox{for all }A\in\mathcal{A}_{\rm inv}\\
&\Leftrightarrow& \tr(A W \rho W^\dagger)=\tr(A W\sigma W^\dagger)\quad\mbox{for all }A\in\mathcal{A}_{\rm inv}\\
&\Leftrightarrow& W\rho W^\dagger\sim W\sigma W^\dagger.
\end{eqnarray*}
Thus, $W$ is OE-preserving. 

Conversely, suppose that $W$ is any OE-preserving unitary. Let $|\psi\rangle\langle\psi|\in\mathcal{A}_{\rm inv}$, and suppose that $\rho$ is some state with $|\psi\rangle\langle\psi|\sim\rho$. Then $\Pi_{\rm inv}(\rho)=\Pi_{\rm inv}(|\psi\rangle\langle\psi|)=|\psi\rangle\langle\psi|$, hence $\rho=|\psi\rangle\langle\psi|$ since $\ket{\psi}\!\bra{\psi}$ is pure. In other words, $|\psi\rangle\langle\psi|$ is only observationally equivalent to itself. Let $\sigma$ be any state with $W^\dagger|\psi\rangle\langle\psi|W\sim\sigma$, then $|\psi\rangle\langle\psi|\sim W\sigma W^\dagger$, and so $W\sigma W^\dagger=|\psi\rangle\langle\psi|$, hence $W^\dagger|\psi\rangle\langle\psi|W$ is only observationally equivalent to itself. But $\tau:=\Pi_{\rm inv}(W^\dagger|\psi\rangle\langle\psi|W)$ satisfies $\Pi_{\rm inv}(\tau)=\Pi_{\rm inv}(W^\dagger|\psi\rangle\langle\psi|W)$, i.e.\ $\tau\sim W^\dagger|\psi\rangle\langle\psi|W$, hence $\tau= W^\dagger|\psi\rangle\langle\psi|W$, and so $W^\dagger|\psi\rangle\langle\psi|W\in\mathcal{A}_{\rm inv}$. Since the pure states of $\mathcal{A}_{\rm inv}$ linearly span $\mathcal{A}_{\rm inv}$, it follows that $W^\dagger\mathcal{A}_{\rm inv} W =\mathcal{A}_{\rm inv}$.

To prove the remaining non-implication, we give a concrete example of a unitary $W$ that preserves $\mathcal{A}_{\rm inv}$ but that is not SE-preserving. To this end, consider two characters $\chi_0,\chi_1\in\mathcal{\hat G}\setminus\{\mathbf{1}\}$ such that $\chi_1$ takes some value which is not taken by $\chi_0$, i.e.\ there exists some $g\in\mathcal{G}$ such that $\chi_1(g)\neq \chi_0(g')$ for all $g'\in\mathcal{G}$. Such pairs of characters exist for many finite Abelian groups (for example for $\mathbb{Z}_4$ by direct inspection of the character table), but not for all of them -- however, we only want to disprove \emph{general} implication for \emph{all} finite Abelian groups, hence we can assume that we have a finite group for which such a pair of characters exists. Let us define some $U\in\mathcal{U}_{\rm sym}$ by picking some assignment $\mathbf{h}\mapsto g(\mathbf{h})$ of group elements $g(\mathbf{h})\in\mathcal{G}$ to particle relations $\mathbf{h}\in\mathcal{G}^{N-1}$. We will pick a mostly arbitrary assignment, except that we demand that there is some $\mathbf{h_0}$ for which $g(\mathbf{h_0}):=g$.

Let us describe states $\rho$ on $\mathcal{H}^{\otimes N}$ via their matrix elements $\rho_{(\mathbf{h},\chi),(\mathbf{h}',\chi')}:=\langle\mathbf{h};\chi|\rho|\mathbf{h}';\chi'\rangle$. Pick any state $\rho$ with $\rho_{(\mathbf{h}_0,\chi_1),(\mathbf{h}_0,\mathbf{1})}\neq 0$. Define $\sigma:=U\rho U^\dagger$, then $\rho\simeq\sigma$. In particular,
\begin{equation}
\sigma_{(\mathbf{h},\chi),(\mathbf{h}',\chi')}=\chi(g(\mathbf{h}))\chi'(g(\mathbf{h}')^{-1})\rho_{(\mathbf{h},\chi),(\mathbf{h}',\chi')},
\label{eqNecSuff}
\end{equation}
and Eq.~(\ref{eqNecSuff}) is a necessary and sufficient condition in general for $\rho\simeq\sigma$: for two states $\rho$ and $\sigma$, we have $\rho\simeq\sigma$ if and only if there is some assignment $\mathbf{h}\mapsto g(\mathbf{h})$ such that Eq.~(\ref{eqNecSuff}) holds. In particular, it holds for our specific choice of $\rho$ and $U$, implying
\begin{equation}
   \sigma_{(\mathbf{h}_0,\chi_1),(\mathbf{h}_0,\mathbf{1})}=\chi_1(g)\rho_{(\mathbf{h}_0,\chi_1),(\mathbf{h}_0,\mathbf{1})}.
   \label{eqContra}
\end{equation}
Now define a unitary $W$ via
\[
   W|\mathbf{h};\chi\rangle:=\left\{
      \begin{array}{cl}
      |\mathbf{h};\chi\rangle& \mbox{if }\chi\not\in\{\chi_0,\chi_1\},\\
      |\mathbf{h};\chi_1\rangle & \mbox{if }\chi=\chi_0\\
      |\mathbf{h};\chi_0\rangle & \mbox{if }\chi=\chi_1.
      \end{array}
   \right.
\]
Due to the form of $\mathcal{A}_{\rm inv}$, it is clear that $W^\dagger \mathcal{A}_{\rm inv} W =\mathcal{A}_{\rm inv}$. However, we will now show that $W$ cannot be SE-preserving. Suppose that it was SE-preserving, then $W\rho W^\dagger\simeq W\sigma W^\dagger$, and the matrix elements of those states would have to satisfy Eq.~(\ref{eqNecSuff}), with $g(\mathbf{h})$ replaced by some other assignment $g'(\mathbf{h})$. In particular,
\begin{eqnarray*}
\sigma_{(\mathbf{h}_0,\chi_1),(\mathbf{h}_0,\mathbf{1})}&=& (W\sigma W^\dagger)_{(\mathbf{h}_0,\chi_0),(\mathbf{h}_0,\mathbf{1})}\\
&=& \chi_0(g'(\mathbf{h}_0))(W\rho W^\dagger)_{(\mathbf{h}_0,\chi_0),(\mathbf{h}_0,\mathbf{1})}\\
&=& \chi_0(g'(\mathbf{h}_0))\rho_{(\mathbf{h}_0,\chi_1),(\mathbf{h}_0,\mathbf{1})}.
\end{eqnarray*}
But together with Eq.~(\ref{eqContra}) this implies that $\chi_1(g)=\chi_0(g'(\mathbf{h}_0))$ which is impossible. Hence $W$ is not SE-preserving.
\end{proof}
In particular, we see that OE-preservation is in general strictly weaker than SE-preservation. Hence, in the following, we will mainly be interested in the former.

\subsection{Perspective-neutral picture: interacting particles}
\label{ssec_interac}
In the perspective-neutral picture, the $\mathcal{G}$-system is in some state of $\mathcal{H}_{\rm phys}$, and every unitary $W\in\mathcal{A}_{\rm phys}$ describes a possible time evolution. This is because $\mathcal{A}_{\rm phys}\subset \mathcal{A}_{\rm inv}$, and so Lemma~\ref{LemChain} guarantees that $W$ preserves both symmetry and observational equivalence of states. Moreover, we can use the Schr\"odinger reduction map to describe this time evolution relative to some particle $i$.

To illustrate this, consider the following concrete example which is a discrete version of the model considered in Ref.~\cite{Vanrietvelde}. As in Example~\ref{ExCyclic} and Figure~\ref{fig_cyclic}, we choose $\mathcal{G}=\mathbb{Z}_n$, i.e.\ the case of $N$ particles on a circle of $n\geq 2$ discrete positions. We define the $N$-particle Hamiltonian
\begin{equation}
   H:=\sum_{i=1}^N \frac{P_i^2 \enspace {\rm mod}\enspace  n}{2 m_i} + \sum_{i<j} V_{i,j},
   \label{eqHamiltonian}
\end{equation}
with $m_i>0$ interpreted as particle masses, and the operators defined as follows. The potential $V_{i,j}$ is the operator
\begin{equation}
   V_{i,j}|g_1,g_2,\ldots,g_N\rangle:=v_{i,j}(g_i^{-1}g_j) |g_1,g_2,\ldots,g_N\rangle,
   \label{DefVij}
\end{equation}
where $v_{i,j}:\mathcal{G}\to\mathbb{R}$ is a real function on the group, and we demand that $v_{i,j}(g^{-1})=v_{i,j}(g)$. Since $g_i^{-1}g_j=g_j-g_i\enspace {\rm mod}\enspace n$, this potential depends only on the discrete distance between particles $i$ and $j$.

To define the momentum operators $P_i$, let us begin with the case of a single particle. Any momentum operator $P'$ should generate translations, i.e.\ should satisfy\footnote{The choice of sign differs from the usual convention in the continuous case. This will not change the physics, but it will slightly simplify the notation. The standard convention $\exp(-\frac{2\pi i}n\ell P)|g\rangle=|g+\ell\rangle$ can be reproduced by setting $P:=\sum_{k=0}^{n-1} k|\overline{\chi_k}\rangle\langle\overline{\chi_k}|$.}
\be
   e^{\frac{2\pi i} n \ell P'} |g\rangle=|g\ell\rangle\equiv|g+\ell\enspace {\rm mod}\enspace n\rangle.
   \label{eqMomentum}
\ee
Modifying the eigenvalues of $P'$ by adding multiples of $n$ will leave this equation invariant. More generally, if $P'$ satisfies Eq.~(\ref{eqMomentum}) then $P'\enspace {\rm mod}\enspace n=P$ (in the sense of spectral calculus), where
\[
   P=\sum_{k=0}^{n-1}k|\chi_k\rangle\langle\chi_k|,
\]
which we thus choose as our definition of the single-particle momentum operator. Here, the vector $|\chi_k\rangle\in\ch$ is defined as $|\mathbf{h};\chi_k\rangle$ for empty $\mathbf{h}$, i.e.\ $\ket{\chi_k}=\frac 1 {\sqrt{n}} \sum_{g\in\mathbb{Z}_n} \chi_k(g^{-1})\ket{g}$. The momentum operator for particle $i$ becomes $P_i:=\mathbf{1}_{\overline i}\otimes P$. In this convention, the eigenvalues of the momentum operator are $\{0,1,\ldots,n-1\}$, which is the same as those of the position operator $X:=\sum_{g\in\mathcal{G}} g |g\rangle\langle g|$, which is related to $P$ by a discrete Fourier transform.

Had we chosen another $P'$, then in general we would have $P^2\neq (P')^2$, but $(P')^2\enspace {\rm mod}\enspace n=P^2\enspace {\rm mod}\enspace n$ would still hold since $P$ and $P'$ have integer eigenvalues. This motivates the ``${\rm mod}\enspace n$''-terms in the definition of our Hamiltonian $H$, making it independent of the choice of momentum operator.

Our first observation is that $H$ does \emph{not} in general generate symmetry-preserving time evolution.
\begin{lemma}
\label{LemIllegalDynamics}
Fix any choice of potentials $V_{i,j}$ and of number of sites $n\geq 2$. Then there exist times $t\in\mathbb{R}$ and values of the masses $m_i$ such that the evolution
\[
   W(t):=e^{-itH}
\]
does \emph{not} preserve observational equivalence.
\end{lemma}
\begin{proof}
We argue by contradiction. Suppose that $W(t)$ is OE-preserving for all times $t\in\mathbb{R}$ and all masses $m_i>0$. According to Lemma~\ref{LemChain}, this implies $W(t)^\dagger A W(t)\in\mathcal{A}_{\rm inv}$ for all $A\in\mathcal{A}_{\rm inv}$. Differentiating this at $t=0$, we obtain
\begin{equation}
   [H,A]\in\mathcal{A}_{\rm inv}\mbox{ for all }A\in\mathcal{A}_{\rm inv}.
   \label{eqCommutator}
\end{equation}
For the potentials $V_{i,j}$, we have
\begin{equation}
   V_{i,j}|\mathbf{h};\chi_k\rangle = v_{i,j}(h_{i-1}^{-1} h_{j-1}) |\mathbf{h};\chi_k\rangle,
   \label{eqPotentialPhys}
\end{equation}
where, as always, we set $h_0:=e=0$, the unit element of the group. Being diagonal in the $|\mathbf{h};\chi\rangle$ basis, this implies that $V_{i,j}\in\mathcal{A}_{\rm inv}$, and so $[V_{i,j},A]\in\mathcal{A}_{\rm inv}$ for all $A\in\mathcal{A}_{\rm inv}$.

Since $H$ depends continuously on $1/m_i$, consider the $H$ which arises from taking the limit $m_i\to\infty$ for all $i\geq 2$. By continuity, condition~(\ref{eqCommutator}) must still hold for the limiting $H$. In particular, let $k\in\{1,2,\ldots,n-1\}$, $\mathbf{h}\in\mathcal{G}^{N-1}$, and $A:=|\mathbf{h};\chi_k\rangle\langle\mathbf{h};\chi_k|$, then we must have $Y:=[P_1^2\enspace {\rm mod}\enspace n,A]\in\mathcal{A}_{\rm inv}$. Now set $\mathbf{h}:=\mathbf{j}+\delta\mbox{ mod }n$, where $\delta\in\{1,\ldots,n-1\}$, then $\mathbf{j}\neq\mathbf{h}$ and $k\neq 0$ imply $\langle \mathbf{j};\chi_k|Y|\mathbf{h};\chi_k\rangle=0$. But by the definition of $P_i$,
\begin{equation}
   P_i^2\enspace {\rm mod}\enspace n =\frac 1 n \sum_{m, h, \ell} (m^2\enspace {\rm mod}\enspace n)e^{\frac{2\pi i m (\ell-h)}{n}}|h\rangle\langle\ell|_i\otimes\mathbf{1}_{\overline i}.
   \label{DefPi}
\end{equation}
Substituting this, we obtain
\begin{eqnarray*}
   \frac n {\chi_k(\delta)}\langle \mathbf{j};\chi_k|Y|\mathbf{h};\chi_k\rangle&=&\frac n {\chi_k(\delta)}\langle \mathbf{j};\chi_k|(P_1^2\enspace {\rm mod}\enspace n)|\mathbf{h};\chi_k\rangle\\
   &=&\sum_{m=0}^{n-1} e^{-\frac{2\pi i m \delta}{n}}(m^2\mbox{ mod } n)=:\hat f(\delta),
\end{eqnarray*}
where $\hat f$ is the discrete Fourier transform of $f(m):=m^2\mbox{ mod }n$. We know that $\hat f(\delta)=0$ for all $\delta\in\{1,\ldots,n-1\}$. On the other hand, for all $m\in\{0,\ldots,n-1\}$,
\[
f(m)= m^2 \mbox{ mod }n =\frac 1 n \sum_{\delta=0}^{n-1} e^{\frac{2\pi i m\delta}n}\hat f(\delta)=\frac{\hat f(0)}n
\]
which is a contradiction.
\end{proof}
Thus, assuming that the $\mathcal{G}$-system evolves according to the Hamiltonian $H$ contradicts our symmetry assumptions that constitute the very foundation of the definition of a $\mathcal{G}$-system. However, the projection of $H$ into the physical subspace
\[
   H_{\rm phys}:=\Pi_{\rm phys} H \Pi_{\rm phys}
\]
\emph{does} describe valid time evolution: since $W_{\rm phys}(t):=e^{-i t H_{\rm phys}} \in\mathcal{A}_{\rm inv}$, it preserves both observational and symmetry equivalence. But why should we be interested in this projected Hamiltonian? The following lemma shows the motivation for doing so:
\begin{lemma}
\label{LemPhysDynamics}
Time evolution $W(t)$ preserves the physical subspace, hence
\[
   \Pi_{\rm phys} W(t)|\psi\rangle = W_{\rm phys}(t)\Pi_{\rm phys}|\psi\rangle
\]
for all $|\psi\rangle\in\mathcal{H}^{\otimes N}$.
\end{lemma}
This result can be interpreted as follows. Fix any initial state $|\psi\rangle\in\mathcal{H}^{\otimes N}$ (for example some alignable state), and consider its ``illegal'' (non-symmetry-preserving) evolution according to the Hamiltonian $H$. Suppose we are not interested in the full evolved state $|\psi(t)\rangle:=W(t)|\psi\rangle$, but only in its ``relational part'', $|\psi_{\rm phys}(t)\rangle:=\Pi_{\rm phys}|\psi(t)\rangle$. Then \emph{this} can be seen as the result of a ``legal'' (symmetry-preserving) time evolution $W_{\rm phys}(t)$ applied to the corresponding relational initial state $|\psi_{\rm phys}\rangle:=\Pi_{\rm phys}|\psi\rangle$.

In other words: $H_{\rm phys}$ tells us how $H$ evolves the \emph{relational part} of the states; and in contrast to $H$, it generates valid \emph{symmetry-preserving dynamics}.
\begin{proof}
It follows directly from the definition~(\ref{DefVij}) of the $V_{i,j}$ that $[V_{i,j},U_g^{\otimes N}]=0$ for all $g$. Similarly, $[P_i^2\enspace {\rm mod}\enspace n,U_g^{\otimes N}]=0$ for all $g$, which can be seen by conjugating Eq.~(\ref{DefPi}) with $U_g^{\otimes N}\bullet (U_g^{\otimes N})^\dagger$. Hence, our Hamiltonian commutes with all global translations,
\[
   [H,U_g^{\otimes N}]=0\mbox{ for all }g\in\mathcal{G},
\]
a result that is very intuitive given its physical interpretation. Due to Eq.~(\ref{eqPiPhys}), we also have
\[
   [H,\Pi_{\rm phys}]=\frac 1 n \sum_g [H,U_g^{\otimes N}]=0,
\]
implying that $[W(t),\Pi_{\rm phys}]=0$: the time evolution preserves the physical subspace. Furthermore, note that $[W_{\rm phys}(t),\Pi_{\rm phys}]=0$, and thus
\begin{eqnarray*}
\Pi_{\rm phys}W(t)&=&\Pi_{\rm phys}W(t)\Pi_{\rm phys}=\Pi_{\rm phys}W_{\rm phys}(t)\Pi_{\rm phys}\\
&=& W_{\rm phys}(t)\Pi_{\rm phys}
\end{eqnarray*}
which proves the statement of the lemma.
\end{proof}
Thus, it is very well-motivated to consider the projection $H_{\rm phys}$ of $H$ to the physical subspace, and to interpret it as the perspective-neutral picture of the $N$-particle evolution.

Lemmas~\ref{LemIllegalDynamics} and~\ref{LemPhysDynamics} have an interesting physical interpretation. The momentum operators $P_i$ (and thus the $P_i^2\mbox{ mod }n$) commute with the total momentum, $P_{\rm tot}:=\left(\sum_{i=1}^N P_i\right)\mbox{mod }n$, which generates global translations: $e^{\frac{2\pi i}n g P_{\rm tot}}=e^{\frac{2\pi i}n g \sum_i P_i}=U_g^{\otimes N}$. In particular, $[P_i,U_g^{\otimes N}]=0$ for all $g\in\mathbb{Z}_n$: momenta are clearly translation-invariant. Since all elements of $\ch_{\rm phys}$ are as well, this shows that the $P_i$ preserve the physical subspace, which is ultimately the reason for the validity of Lemma~\ref{LemPhysDynamics}: our Hamiltonian generates valid time evolution when restricted to $\ch_{\rm phys}$.

On the other hand, our $\cg$-systems have a larger symmetry group $\mathcal{U}_{\rm sym}$ of \emph{relation-conditional translations}. For example, for some pair of particles $j\neq k$, consider the operator $\Delta_{j,k}:=X_j-X_k\mbox{ mod }n$. It satisfies $\Delta_{j,k}|\mathbf{g}\rangle=g_j g_k^{-1}|\mathbf{g}\rangle\equiv (g_j -g_k \mbox{ mod }n)|\mathbf{g}\rangle$, i.e.\ it determines the `distance' between particles $j$ and $k$. Thus, it commutes with the global translations $U_g^{\otimes N}$, in particular for $g=1$. Let us choose a convention for the complex logarithm by defining $\log(r e^{ix}):=\log r +ix$ whenever $r>0$ and $0\leq x <2\pi$. Then $P_{\rm tot}=\frac n {2\pi i}\log U_1^{\otimes N}$, thus $[\Delta_{j,k},P_{\rm tot}]=0$. Now define the unitary $U:=e^{\frac{2\pi i}n \Delta_{j,k} P_{\rm tot}}$. It acts as a global translation by an amount that depends on the distance between the particles $j$ and $k$:
\[
   e^{\frac{2\pi i}n \Delta_{j,k}P_{\rm tot}}|g,\mathbf{h}g\rangle=U_{h_{j-1}-h_{k-1}\mathrm{ mod }\, n}^{\otimes N}|g,\mathbf{h}g\rangle,
\]
thus $U\in\mathcal{U}_{\rm sym}$. But while the \emph{total} momentum commutes with $U$, the \emph{individual momenta} do \emph{not} commute with $U$ in general. To see this, consider the special case $N=2$ and $(j,k)=(2,1)$. Suppose that $[P_1,U]=0$. This implies that $[U_g\otimes\mathbf{1}_2,U]=0$ for all $g$, and in particular for $g=1$, since $U_g=e^{\frac{2\pi i}n g P_1}$. But it is easy to see that
\begin{eqnarray*}
(U_1\otimes \mathbf{1}_2)U|g,g+h\rangle&=&|g+h+1,g+2h\rangle,\\
U(U_1\otimes\mathbf{1}_2)|g,g+h\rangle&=&|g+h,g+2h-1\rangle,
\end{eqnarray*}
where all additions are modulo $n$. This is a contradiction. We hence see that the momenta are not invariant observables, $P_i\not\in\ca_{\rm inv}$. If they were, then we would have $W(t)=e^{-itH}\in\ca_{\rm inv}$ as well, but then Lemma~\ref{LemChain} would guarantee that $W(t)$ preserves observational equivalence, which it does not as we have shown in Lemma~\ref{LemIllegalDynamics}.

Ultimately, demanding that an operator is invariant under all $U\in\mathcal{U}_{\rm sym}$ is a much stricter requirement than unconditional global translation-invariance, and the $P_i$ do not satisfy this stricter requirement. Consequently, time evolution according to $H$, considered on \emph{all} of $\ch^{\otimes N}$, must violate some of the symmetries in $\mathcal{U}_{\rm sym}$.

In Subsection~\ref{SubsecReducing}, we have seen that we can reduce the perspective-neutral description to a frame perspective: for example, we can express what a state or evolution looks like ``relative to the first particle''. Let us do this now for our Hamiltonian~(\ref{eqHamiltonian}).

Suppose we are given an operator $X$ on $\mathcal{H}_{\rm phys}$, which is spanned by the basis vectors $|\mathbf{h};\mathbf{1}\rangle$. We have chosen the convention for labelling these basis vectors in such a way that $\mathbf{h}$ denotes the relations of all other particles \emph{relative to the first particle}. But relative to the first particle, the configuration with these relations is exactly $|e,\mathbf{h}\rangle$. Thus, the description of $X$ relative to the first particle should correspond to an operator $\tilde X$ on $\mathcal{\tilde H}_{\overline 1}:={\rm span}\{|e,\mathbf{h}\rangle,\,|\,\,\mathbf{h}\in\mathcal{G}^{N-1}\}=|e\rangle\otimes\mathcal{H}_{\overline 1}$ such that
\begin{equation}
   \langle\mathbf{j};\mathbf{1}|X|\mathbf{h};\mathbf{1}\rangle=\langle e,\mathbf{j}|\tilde X|e,\mathbf{h}\rangle.
   \label{eqMatrixElements}
\end{equation}
Recall our reduction map $\mathcal{\overline R}=\ket{e}_1\otimes\mathcal{R}$ from Subsection~\ref{SubsecPageWootters}. We have seen that it satisfies $\mathcal{R}|\mathbf{h};\mathbf{1}\rangle=|\mathbf{h}\rangle$. We find that
\[
   \tilde X = \mathcal{\overline R}X \mathcal{\overline R}^\dagger.
\]
Before we can apply this transformation to our Hamiltonian $H$, we need a lemma that tells us how to apply functions to operators under this transformation:
\begin{lemma}
\label{LemSpectralCalculus}
Let $A$ and $B$ be normal operators on $\mathcal{H}^{\otimes N}$ with $[A,\Pi_{\rm phys}]=[B,\Pi_{\overline 1}]=0$, where $\Pi_{\overline 1}:=|e\rangle\langle e|_1\otimes\mathbf{1}_{\overline 1}$ is the projection onto $\mathcal{\tilde H}_{\overline 1}$. If for all $\mathbf{j},\mathbf{h}\in\mathcal{G}^{N-1}$
\begin{equation}
   \langle \mathbf{j};\mathbf{1}|A|\mathbf{h};\mathbf{1}\rangle=\langle e,\mathbf{j}|B|e,\mathbf{h}\rangle,
   \label{eqME1}
\end{equation}
then also
\begin{equation}
   \langle \mathbf{j};\mathbf{1}|f(A)|\mathbf{h};\mathbf{1}\rangle=\langle e,\mathbf{j}|f(B)|e,\mathbf{h}\rangle
   \label{eqME2}
\end{equation}
for all functions $f:\mathbb{C}\to\mathbb{C}$ applied in the sense of spectral calculus.
\end{lemma}
\begin{proof}
The commutation relations $[A,\Pi_{\rm phys}]=[B,\Pi_{\overline 1}]=0$ imply that $A$ and $B$ are block-diagonal. Furthermore, Eq.~(\ref{eqME1}) implies that the blocks $\left. A \right|_{\mathcal H_{\mathrm{phys}}}$ and $\left. B \right|_{\tilde{\mathcal{H}}_{\overline{1}}}$ are identical, up to the relabeling $\ket{\mathbf{h} ; \mathbf{1} } \mapsto \ket{e, \mathbf{h}}$. Therefore, also the spectral decompositions of these two blocks are identical (up to the relabeling). Thus, also the blocks $\left. f(A) \right|_{\mathcal H_{\mathrm{phys}}}$ and $\left. f(B) \right|_{\tilde{\mathcal{H}}_{\overline{1}}}$ are identical, up to the relabeling $\ket{\mathbf{h} ; \mathbf{1}} \mapsto \ket{e, \mathbf{h}}$. But that is precisely the statement of Eq.~(\ref{eqME2}).
\end{proof}
Applying this transformation to our Hamiltonian $H$, we obtain the following result.
\begin{lemma}
The $N$-particle Hamiltonian~(\ref{eqHamiltonian}) as expressed relative to the first particle is
\begin{eqnarray*}
   \tilde H &=& \sum_{i=2}^N \frac{P_i^2\enspace{\rm mod}\enspace n}{2 m_i}+\sum_{i<j}V_{i,j}\\
   &&+\frac 1 {2 m_1} \left[\left(\sum_{i=2}^N P_i^2+\sum_{2\leq i<j}P_i P_j\right){\rm mod}\enspace n\right].
\end{eqnarray*}
\end{lemma}
Before turning to the proof, a few words of clarification are in place. The operator $\tilde H$ is by definition an operator on $\mathcal{\tilde H}_{\overline 1}$, whereas the operators $P_i^2$ and $V_{i,j}$ appearing in the lemma are defined on all of $\mathcal{H}^{\otimes N}$. Hence, for a more mathematically precise formulation of the lemma, one should write $\Pi_{\overline 1} V_{i,j}\Pi_{\overline 1}$ instead of $V_{i,j}$ and $\Pi_{\overline 1} (P_i^2 \enspace{\rm mod}\enspace n)\Pi_{\overline 1}$ instead of $P_i^2 \enspace{\rm mod}\enspace n$.

Note that the continuous version of this transformed Hamiltonian has already been described in Ref.~\cite{Vanrietvelde}, and in the special case $V_{i,j}=0$ in Ref.~\cite{Aharonov3}.
\begin{proof}
The previous comment also applies to the proof: we will sometimes omit the projectors. For example, we will now first show that
\[
   \tilde V_{i,j}=V_{i,j} \mbox{ for all }i<j,
\]
by which we mean that the relation~(\ref{eqME1}) holds for $A=B=V_{i,j}$. Our first claim follows by noting that
\[
   V_{i,j}|e,\mathbf{h}\rangle=v_{i,j}(h_{i-1}^{-1} h_{j-1})|e,\mathbf{h}\rangle
\]
and comparing with Eq.~(\ref{eqPotentialPhys}).

Next, we claim that
\begin{equation}
   \tilde P_i=P_i\mbox{ for all }i\geq 2,
   \label{eqPi2}
\end{equation}
and thus also $\tilde P_i^2\mbox{ mod }n=P_i^2\mbox{ mod }n$ according to Lemma~\ref{LemSpectralCalculus}. This can be seen by directly comparing the matrix elements,
\begin{eqnarray*}
\langle e,\mathbf{j}|P_i|e,\mathbf{h}\rangle &=&  \frac 1 n \sum_k k e^{2 \pi i (h_{i-1}-j_{i-1})k/n} \prod_{\ell\neq i-1} \delta_{h_\ell,j_\ell}\\
&=&\langle \mathbf{j};\mathbf{1}|P_i|\mathbf{h};\mathbf{1}\rangle\qquad (i\geq 2).
\end{eqnarray*}
However, there is a more elegant argument. Consider the shift operator $T_i:=e^{\frac{2 \pi i}n P_i}$. We have $T_i|e,\mathbf{h}\rangle=|e,\mathbf{j}\rangle$, where $\mathbf{j}$ is the same as $\mathbf{h}$, except that the $(i-1)$th entry has been increased by $1$. Therefore $\tilde T_i=T_i$. Moreover, $[T_i,\Pi_{\rm phys}]=0$ and, since $i\geq 2$, also $[T_i,\Pi_{\overline 1}]=0$. We have $P_i=\frac n {2\pi i}\log T_i$, and thus $\tilde P_i=P_i$ for all $i\geq 2$.

Finally, consider $T_1:=e^{\frac{2\pi i}n P_1}$. Direct calculation shows that
\[
   T_1 |\mathbf{h};\mathbf{1}\rangle=|\mathbf{h}-1;\mathbf{1}\rangle,
\]
i.e.\ each entry of $\mathbf{h}$ is shifted by negative one, modulo $n$. A corresponding operator doing this on $\mathcal{\tilde H}_{\overline 1}$ is given by
\[
   T_2^{-1}T_3^{-1}\ldots T_N^{-1}|e,\mathbf{h}\rangle=|e,\mathbf{h}-1\rangle.
\]
Thus, we can choose $\tilde T_1=T_2^{-1}T_3^{-1}\ldots T_N^{-1}=e^{\frac{2 \pi i} n (-P_2 - P_3-\ldots-P_N)}$ such that Eq.~(\ref{eqMatrixElements}) holds. Now, since $P_1=\frac n {2\pi i} \log T_1$ according to our definition of $\log$, Lemma~\ref{LemSpectralCalculus} implies that we can choose $\tilde P_1=\frac n {2\pi i}\log e^{\frac{2\pi i} n (-P_2-P_3-\ldots-P_N)}$. But this equals
\[
   \tilde P_1=(-P_2-P_3-\ldots-P_N)\mbox{ mod }n.
\]
Note that the operator in brackets has only integer eigenvalues. Since for integer $s\in\mathbb{Z}$, we have
\[
   (s\mbox{ mod }n)^2\mbox{ mod }n=s^2\mbox{ mod }n,
\]
we obtain via Lemma~\ref{LemSpectralCalculus}
\begin{eqnarray*}
   \widetilde{P_1^2\mbox{ mod }n}&=&\tilde P_1^2 \mbox{ mod }n\\
   &=&(-P_2-P_3-\ldots-P_N\mbox{ mod }n)^2\mbox{ mod }n\\
   &=& (P_2+P_3+\ldots+P_N)^2\mbox{ mod }n\\
   &=& \left(P_2^2+\ldots+P_N^2+\sum_{2\leq i <j}P_i P_j\right)\mbox{ mod }n.
\end{eqnarray*}
This proves the claimed form of $\tilde H$.
\end{proof}
An interesting upshot of this lemma is that relative to the first (and therefore any) particle, the Hamiltonian will \emph{always} feature interaction terms between the remaining particles in $\overline{1}$ due to the $P_iP_j$-terms arising through solving the discrete total momentum constraint. This is in particular the case when $V_{i,j}=0$ for all $i,j$, so that the total Hamiltonian $H$ of the perspective-neutral level in Eq.~\eqref{eqHamiltonian} is free.

\subsection{Time evolution in the alignability picture}
In this picture, the $\mathcal{G}$-system is in some alignable state $|\psi\rangle$, which is dynamically equivalent to $|e\rangle_1\otimes |\varphi\rangle_{\overline 1}$ for some $|\varphi\rangle\in\mathcal{H}^{\otimes (N-1)}$. What are the possible time evolutions ``relative to the first particle''? Such time evolution should be described by a global unitary $W$ which acts as\footnote{In principle, it would be sufficient to demand symmetry-equivalence except of identity, but we can always absorb the corresponding symmetry $U\in\mathcal{U}_{\rm sym}$ into the definition of $W$.}
\begin{equation}
   W\left(|e\rangle_1\otimes|\varphi\rangle_{\overline 1}\right)=|e\rangle_1\otimes W^{\overline 1}|\varphi\rangle_{\overline 1}.
   \label{eqEvolution}
\end{equation}
Which unitaries $W^{\overline 1}$ can be implemented if we demand that $W$ preserves (at least observational) equivalence? The discussion around Corollary~\ref{cor_nonlinalign} already indicates that there cannot be many such unitaries. To answer this question in more detail, we need the following lemma:
\begin{lemma}
\label{LemLinAlg}
If $W^\dagger \mathcal{A}_{\rm inv} W =\mathcal{A}_{\rm inv}$ then $\Pi_{\rm inv}(W\rho W^\dagger)=W\Pi_{\rm inv}(\rho)W^\dagger$.
\end{lemma}
\begin{proof}
This can be understood via simple linear algebra: with respect to the Hilbert-Schmidt inner product $\langle X,Y\rangle:=\tr(X^\dagger Y)$, $\Pi_{\rm inv}$ is the orthogonal projection onto $\mathcal{A}_{\rm inv}$, and $W^\dagger \bullet W$ is an orthogonal map that	preserves $\mathcal{A}_{\rm inv}$, hence applying it before or after the projection does not make any difference. More formally, we have $\langle W^\dagger X W,Y\rangle=\langle X,W Y W^\dagger\rangle$, and if $\{X_i\}_{i=1}^{\dim \mathcal{A}_{\rm inv}}$ is any orthonormal basis of $\mathcal{A}_{\rm inv}$, then $\{W^\dagger X_i W\}_{i=1}^{\dim \mathcal{A}_{\rm inv}}$ is another orthonormal basis of $\mathcal{A}_{\rm inv}$, and so
\begin{eqnarray*}
\Pi_{\rm inv}(W\rho W^\dagger)&=&\sum_i \langle X_i,W\rho W^\dagger\rangle X_i\\
&=&W\left(\sum_i \langle W^\dagger X_i W,\rho\rangle W^\dagger X_i W\right)W^\dagger\\
&=& W \Pi_{\rm inv}(\rho)W^\dagger,
\end{eqnarray*}
proving the claim.
\end{proof}
This allows us to answer the question above:
\begin{theorem}
Let $W^{\overline 1}$ be some unitary on $\mathcal{H}_{\overline 1}$ which is implemented by an OE-preserving unitary $W$ as in Eq.~(\ref{eqEvolution}). Then there are $d:=|\mathcal{G}|^{N-1}$ complex phases $e^{i\theta_1},\ldots,e^{i\theta_d}$ and a $d\times d$ permutation matrix $P$ such that, in $\{|\mathbf{h}\rangle\}$-basis,
\[
   W^{\overline 1}=P\left(\begin{array}{ccc} e^{i\theta_1} & & \\ & \ddots & \\ & & e^{i \theta_d}\end{array}\right).
\]
In particular, $W^{\overline 1}$ cannot create any superpositions of classical configurations $|\mathbf{h}\rangle$, and the only Hamiltonians generating time-continuous evolution of this kind are of the form
\[
   H=\sum_{\mathbf{h}\in\mathcal{G}^{N-1}} E_{\mathbf{h}} |\mathbf{h}\rangle\langle\mathbf{h}| \qquad (E_{\mathbf{h}}\in\mathbb{R}).
\]
\end{theorem}
\textit{Remark.} The theorem also holds if we alternatively assume that $W$ is SE-preserving: due to Lemma~\ref{LemChain}, this is an even stronger condition than OE-preservation.
\begin{proof}
Define $w_{\mathbf{h},\mathbf{j}}:=\langle\mathbf{h}|W^{\overline 1}|\mathbf{j}\rangle$, then
\[
   W\left(|e\rangle\langle e|_1\otimes |\mathbf{m}\rangle\langle\mathbf{m}|_{\overline 1}\right)W^\dagger=
   |e\rangle\langle e|_1\otimes \sum_{\mathbf{h},\mathbf{j}} w_{\mathbf{h},\mathbf{m}} \overline w_{\mathbf{j},\mathbf{m}} |\mathbf{h}\rangle\langle\mathbf{j}|.
\]
According to~\cite[Lemma 22]{QRF1}, applying $\Pi_{\rm inv}$ to this expression yields
\begin{equation}
   \sum_{\mathbf{h},\mathbf{j}} \frac{w_{\mathbf{h},\mathbf{m}} \overline w_{\mathbf{j},\mathbf{m}}}{|\mathcal{G}|}|\mathbf{h};\mathbf{1}\rangle\langle\mathbf{j};\mathbf{1}|+\sum_{\mathbf{h}} \frac{|w_{\mathbf{h},\mathbf{m}}|^2}{|\mathcal{G}|} \Pi_{\mathbf{h};\chi\neq \mathbf{1}}.
   \label{eqCompare1}
\end{equation}
On the other hand, the same lemma implies
\begin{equation}
   \Pi_{\rm inv}\left(|e\rangle\langle e|_1\otimes |\mathbf{m}\rangle\langle\mathbf{m}|_{\overline 1}\right)=\frac 1 {|\mathcal{G}|} \Pi_{\mathbf{m}},
   \label{eqCompare2}
\end{equation}
where $\Pi_{\mathbf{m}}=\sum_{g\in\cg} |g,\mathbf{m}g\rangle\langle g,\mathbf{m}g|$. Since $W$ is OE-preserving, Lemma~\ref{LemChain} tells us that $W^\dagger \mathcal{A}_{\rm inv} W =\mathcal{A}_{\rm inv}$, and thus $\Pi_{\rm inv}(W\rho W^\dagger)=W\Pi_{\rm inv}(\rho)W^\dagger$ for all $\rho$ due to Lemma~\ref{LemLinAlg}. Hence, conjugation of~(\ref{eqCompare2}) with $W$ must give us (\ref{eqCompare1}): in particular, both expressions must have the same sets of eigenvalues. But the eigenvalues of~(\ref{eqCompare2}) are $0$ and $1/|\mathcal{G}|$, and so $|w_{\mathbf{h},\mathbf{m}}|^2\in\{0,1\}$ for all $\mathbf{h}$ and $\mathbf{m}$. Since $W^{\overline 1}$ is a unitary matrix, this means that every column (or row) contains exactly one entry of the form $e^{i\theta}$ for some phase $\theta\in\R$. Hence we can write $W^{\overline 1}$ in the claimed form.
\end{proof}
Thus, we see that time evolution for alignable states is severely limited: if we demand that time evolution preserves the symmetries of the $\cg$-system --- at least in the weak sense of OE-preservation --- then no superpositions of distinct relations $\mathbf{h}$ can ever be created. This is in stark contrast to relational states which allow dynamics corresponding to \emph{all} unitaries on $\ca_{\rm phys}$, admitting non-zero transition amplitudes between all distinct relations $\mathbf{h}$.

\section{Conclusions}
\label{SecConclusions}
We have provided a systematic and fairly complete analysis of quantum reference frames subject to a finite Abelian group symmetry principle, extending our earlier work \cite{QRF1}. The setup is, on the one hand, simple enough to permit a rigorous information-theoretic analysis and, on the other, sufficiently rich so as to explore many of the conceptual and structural conundrums appearing also in setups with more complicated symmetry groups. Finite Abelian groups thereby offer an ideal testbed for quantum reference frame physics and even (lattice) gauge theories. We have taken advantage of this to characterize structures associated with constraint quantization from an information-theoretic perspective that appear in the perspective-neutral approach to quantum frame covariance. We have further revealed the relations between different approaches to QRFs, especially the ``alignability'' and perspective-neutral approach, the latter of which also encompasses the Page-Wootters and relational observable formalisms often invoked in quantum gravity and cosmology. 

In particular, in conjunction with Ref.~\cite{QRF1}, this article offers several compelling arguments in favor of the perspective-neutral physical Hilbert space $\ch_{\rm phys}$ and the associated algebra of $\ca_{\rm phys}$ relational observables:
\begin{itemize}
\item The subspace $\ch_{\rm phys}$ is the maximal subspace of states that can be purified in an external-frame-independent manner.
\item In contrast to the invariant (incoherently group-averaged) algebra $\ca_{\rm inv}$, the (coherently-averaged) algebra of relational observables $\ca_{\rm phys}$ admits an unambiguous and invariant notion of partial trace\,---\,the \emph{relational trace}\,---\,generalizing the standard partial trace to contexts with symmetry. For this reason, the perspective-neutral approach does not feature the `paradox of the third particle' \cite{QRF1}.
\item The set of alignable states admits no non-trivial \emph{symmetry-preserving} unitary dynamics with non-vanishing transition probabilities between distinct subsystem relations. By contrast, $\ch_{\rm phys}$ admits non-trivial transitions between arbitrary subsystem relations. 
\end{itemize}
We can rephrase this finding as follows. If there are actual physical $\cg$-systems in nature, subject to the corresponding symmetry principles but evolving non-trivially in time, then it is inconsistent to assume that these systems are for all times described by alignable states. In contrast, it \emph{is} consistent to assume that these systems are described by relational states of $\ch_{\rm phys}$, and the consistency of this perspective-neutral description is preserved under time evolution and composition.

Several of these observations have their root in the fact that $\ch_{\rm phys}$ is a proper Hilbert space that is automatically invariant under symmetries. As such it comes with a non-trivial set of observables and dynamics with coherence across different subsystem relations. By contrast, the set of invariant, i.e.\ incoherently group-averaged states, while forming a convex set, does not constitute the full set of density matrices over some Hilbert space
and the set of pure alignable states is \emph{not} even a Hilbert space. Lastly, the set of pure states that are already aligned to a given frame \emph{does} form a Hilbert space, however, its unitaries are not part of the symmetry group. This severely restricts the possibilities for observables and transformations that can be applied to invariant, alignable and aligned states, while remaining consistent with the symmetries.

This is also ultimately the reason why the ``alignability'' and perspective-neutral approaches are \emph{kinematically equivalent}, yet \emph{dynamically inequivalent}\,---\,when requiring symmetry-preservation. It is clear from Sections~\ref{SubsecPageWootters} and~\ref{ssec_interac}, however, that the dynamics on perspective-neutral states \emph{is} equivalent to the reduced dynamics on aligned states. The price one pays is that the latter is not a symmetry as observed above; indeed, since it has to leave the reference frame state invariant (e.g., leave the frame particle in the origin) while non-trivially evolving the remaining degrees of freedom, it actually changes the relation between the frame and its complement (which is invariant under symmetries).\footnote{For the reader familiar with gauge theories, note that there is a difference in nomenclature here. What we call symmetries would in gauge theories be called field-dependent gauge transformations, while the transformations changing the relation between the frame and the remaining degrees of freedom are called `symmetries'~\cite{Donnelly:2016auv,CH1}. For consistency with our quantum information based nomenclature in Ref.~\cite{QRF1}, we call these concepts differently here.} This highlights that insisting on symmetry-preservation in the context of alignable and aligned states is an unnecessarily rigid concept. Indeed, since any particle in our model can assume the role of a reference frame, it is clear that one needs to permit the relation between the frame and its complement to change if one is to obtain a non-trivial dynamics that permits transitions between different subsystem relations. 

The results of this paper and our previous publication~\cite{QRF1} can be seen as steps towards an axiomatic and operational approach to relational quantum physics. Such relational ideas are relevant in gauge theories or quantum gravity, where \emph{internal} quantum systems have to be promoted to reference frames in certain contexts. However, the standard methods to do so --- via relational observables \cite{Rovellibook,Rovelli1,Rovelli2,Rovelli3,Dittrich1,Dittrich2,Thiemann,Tambornino,Hoehn:2018aqt,Hoehn:2018whn,Hoehn:2019owq,Hoehn:2020epv,Chataignier,Chataignier2,Chataignier3}, the Page-Wootters formalism~\cite{Page,Hoehn:2019owq,Hoehn:2020epv,Giovanetti,Alex1,Alex3,Castro,Baumann:2021ifs}, or edge modes \cite{Casini:2013rba,Donnelly:2016auv,Geiller:2019bti,Gomes:2018shn,Riello:2021lfl,Wieland:2017zkf,Wieland:2017cmf,Freidel:2020xyx,CH1} --- tend to sweep certain questions under the rug which are center stage in quantum information theory: what is the role of measurement in such approaches? Do these approaches represent the only possibilities to formulate quantum theory relationally? How should we think of an observer that concretely assigns states relative to internal quantum systems? What happens to the \emph{quantum information about other systems} contained in conditional states relative to a quantum rod or clock?

Here and in Ref.~\cite{QRF1}, we have answered several of these questions for the case of finite Abelian symmetry groups. This contributes to the goal of achieving a foundational understanding of relational quantum physics from first principles. It is now well-known that the complete Hilbert space formalism of quantum theory --- in its ordinary formulation that implicitly assumes perfect external rods and clocks --- can be derived from such principles~\cite{LesHouches,HardyAxioms,DakicBrukner,MasanesMueller,Chiribella,HoehnToolbox,BarnumHilgert}. Achieving a similar reconstruction of the most general formulation of relational quantum physics would not only improve our understanding of it, but might also provide us with a clearer operational perspective on the meaning of its formalism in quantum gravity and beyond.

\section*{Acknowledgments}
We would like to thank Anne-Catherine de la Hamette, Stefan Ludescher, Isha Kotecha and Fabio Mele for inspiring discussions. PAH is grateful for support from the Foundational Questions Institute under grant number FQXi-RFP-1801A. MK and MPM acknowledge support from the Austrian Science Fund (FWF) through the project P 33730-N. MK acknowledges financial support from the European Commission via Testing the Large-Scale Limit of Quantum Mechanics (TEQ) (No.\ 766900) project. This work was supported in part by funding from Okinawa Institute of Science and Technology Graduate University. This research was supported in part by Perimeter Institute for Theoretical Physics. Research at Perimeter Institute is supported by the Government of Canada through the Department of Innovation, Science, and Economic Development, and by the Province of Ontario through the Ministry of Colleges and Universities.


\begin{thebibliography}{99}

\bibitem{QRF1}
M.\ Krumm, P.\ A.\ H\"ohn, and M.\ P.\ M\"uller, \emph{Quantum reference frame transformations as symmetries and the paradox of the third particle}, Quantum \textbf{5}, 530 (2021).

\bibitem{Rovellibook}
C.~Rovelli, \emph{Quantum gravity}, Cambridge University Press, 2004.

\bibitem{Thiemann}
T.~Thiemann, \emph{Modern canonical quantum general relativity}, Cambridge University Press, 2007.

\bibitem{Tambornino}
J.~Tambornino,
\emph{Relational Observables in Gravity: a Review,}
SIGMA \textbf{8}, 017 (2012).

\bibitem{Rovelli1}
C.~Rovelli,
\emph{What Is Observable in Classical and Quantum Gravity?},
Class.\ Quant.\ Grav.\ \textbf{8}, 297 (1991).

\bibitem{Rovelli2}
C.~Rovelli,
\emph{Quantum reference systems},
Class.\ Quant.\ Grav.\ \textbf{8}, 317 (1991).

\bibitem{Rovelli3}
C.~Rovelli,
\emph{Time in Quantum Gravity: Physics Beyond the Schr\"odinger Regime},
Phys.\ Rev.\ D \textbf{43}, 442-456 (1991).

\bibitem{Dittrich1}
B.~Dittrich, \emph{Partial and complete observables for Hamiltonian constrained systems}, Gen.\ Rel.\ Grav.\ \textbf{39}, 1891 (2007).

\bibitem{Dittrich2}
B.~Dittrich, \emph{Partial and complete observables for canonical general relativity}, Class.\ Quant.\ Grav.\ \textbf{23}, 6155 (2006).

\bibitem{Aberg}
J.\ \r{A}berg, \emph{Catalytic Coherence}, Phys.\ Rev.\ Lett.\ \textbf{113}, 150402 (2014).

\bibitem{LJR}
M.\ Lostaglio, D.\ Jennings, and T.\ Rudolph, \emph{Description of quantum coherence in thermodynamic processes requires constraints beyond free energy}, Nat.\ Comm.\ \textbf{6}, 6383 (2015).

\bibitem{LKJR}
M.\ Lostaglio, K.\ Korzekwa, D.\ Jennings, and T.\ Rudolph, \emph{Quantum Coherence, Time-Translation Symmetry, and Thermodynamics}, Phys.\ Rev.\ X \textbf{5}, 021001 (2015).

\bibitem{LostaglioMueller}
M.\ Lostaglio and M.\ P.\ M\"uller, \emph{Coherence and asymmetry cannot be broadcast}, Phys.\ Rev.\ Lett.\ \textbf{123}, 020403 (2019).

\bibitem{MarvianSpekkens}
I.\ Marvian and R.\ W.\ Spekkens, \emph{No-Broadcasting Theorem for Quantum Asymmetry and Coherence and a Trade-off Relation for Approximate Broadcasting}, Phys.\ Rev.\ Lett.\ \textbf{123}, 020404 (2019).

\bibitem{Erker}
P.\ Erker, M. T.\ Mitchison, R.\ Silva, M.\ P.\ Woods, N.\ Brunner, and M.\ Huber, \emph{Autonomous quantum clocks: does thermodynamics limit our ability to measure time?}, Phys.\ Rev.\ X \textbf{7}, 031022 (2017).

\bibitem{Cwiklinski}
P.\ \'Cwikli\'nski, M.\ Studzi\'nski, M.\ Horodecki, and J.\ Oppenheim, \emph{Limitations on the Evolution of Quantum Coherences: Towards Fully Quantum Second Laws of Thermodynamics}, Phys.\ Rev.\ Lett.\ \textbf{115}, 210403 (2015).

\bibitem{Woods1}
M~P.~Woods, R.~Silva, and J.~Oppenheim, \emph{Autonomous quantum machines and finite-sized clocks} Ann.\ Henri Poincar\'e \textbf{20}, 125 (2019). 

\bibitem{Woods2}
M.~P.~Woods, and M~Horodecki, \emph{The resource theoretic paradigm of quantum thermodynamics with control}, arXiv:1912.05562 [quant-ph].

\bibitem{Bartlett}
S.\ D.\ Bartlett, T.\ Rudolph, and R.\ W.\ Spekkens, \emph{Reference frames, superselection rules, and quantum information}, Rev.\ Mod.\ Phys.\ \textbf{79}, 555 (2007).

\bibitem{Smith2019}
A.\ R.\ H.\ Smith, \emph{Communicating without shared reference frames}, Phys.\ Rev.\ \textbf{A 99}, 052315 (2019).

\bibitem{Marvian}
I.\ Marvian, \emph{Symmetry, Asymmetry and Quantum Information}, PhD thesis, University of Waterloo, 2012.

\bibitem{ResourceTheoryQRF}
G.\ Gour and R.\ W.\ Spekkens, \emph{The resource theory of quantum reference frames: manipulations and monotones}, New J.\ Phys.\ \textbf{10}, 033023 (2008).

\bibitem{Frameness}
G.\ Gour, I.\ Marvian, and R.\ W.\ Spekkens, \emph{Measuring the quality of a quantum reference frame: The relative entropy of frameness}, Phys.\ Rev.\ A \textbf{80}, 012307 (2009).

\bibitem{Modes}
I.\ Marvian and R.\ W.\ Spekkens, \emph{Modes of asymmetry: The application of harmonic analysis to symmetric quantum dynamics and quantum reference frames}, Phys.\ Rev.\ A \textbf{90}, 062110 (2014).

\bibitem{Palmer}
M.\ C.\ Palmer, F.\ Girelli, and S.\ D.\ Bartlett, \emph{Changing quantum reference frames}, Phys.\ Rev.\ A \textbf{89}, 052121 (2014).

\bibitem{Smith2016}
A.\ R.\ H.\ Smith, M.\ Piani, and R.\ B.\ Mann, \emph{Quantum reference frames associated with noncompact groups: the case of translations and boosts, and the role of mass}, Phys.\ Rev.\ \textbf{A 94}, 012333 (2016).

\bibitem{Aharonov1}
Y.\ Aharonov and L.\ Susskind, \emph{Charge Superselection Rule}, Phys.\ Rev.\ \textbf{155}, 1428 (1967).

\bibitem{Aharonov2}
Y.\ Aharonov and L.\ Susskind, \emph{Observability of the Sign Change of Spinors under $2\pi$ Rotations}, Phys.\ Rev.\ \textbf{158}, 1237 (1967).

\bibitem{Aharonov3}
Y.\ Aharonov and T.\ Kaufherr, \emph{Quantum frames of reference}, Phys.\ Rev.\ D \textbf{30}, 368 (1984).

\bibitem{Wigner}
E.\ Wigner, \emph{Die Messung quantenmechanischer Operatoren}, Z.\ Phys.\ \textbf{133}, 101 (1952).

\bibitem{Araki}
H.\ Araki and M.\ M.\ Yanase, \emph{Measurement of Quantum Mechanical Operators}, Phys.\ Rev.\ \textbf{120}, 622 (1960).

\bibitem{Yanase}
M.\ M.\ Yanase, \emph{Optimal Measuring Apparatus}, Phys.\ Rev.\ \textbf{123}, 666 (1961).

\bibitem{Loveridge2017}
L.\ Loveridge, B.\ Busch, and T.\ Miyadera, \emph{Relativity of quantum states and observables}, EPL \textbf{117}, 40004 (2017).

\bibitem{Loveridge2018}
L.\ Loveridge, T.\ Miyadera, and P.\ Busch, \emph{Symmetry, Reference Frames, and Relational Quantities in Quantum Mechanics}, Found.\ Phys.\ \textbf{48}, 135--198 (2018).

\bibitem{Miyadera}
T.\ Miyadera, L.\ Loveridge, and P.\ Busch,  \emph{Approximating relational observables by absolute quantities: a quantum accuracy-size trade-off}, J.\ Phys. A: Mathematical and Theoretical, \textbf{49}(18), 185301 (2016).

\bibitem{Loveridge2020}
L.~Loveridge, \emph{A relational perspective on the Wigner-Araki-Yanase theorem}, J.\ Phys.: Conf.\ Ser.\ \textbf{1638}, 012009 (2020).

\bibitem{HoehnMueller}
P.\ A.\ H\"ohn and M.\ P.\ M\"uller, \emph{An operational approach to spacetime symmetries: Lorentz transformations from quantum communication}, New J.\ Phys.\ \textbf{18}, 063026 (2016).

\bibitem{Giacomini}
F.\ Giacomini, E.\ Castro-Ruiz, and \v{C}.\ Brukner, \emph{Quantum mechanics and the covariance of physical laws in quantum reference frames}, Nat.\ Comm.\ \textbf{10}, 494 (2019).

\bibitem{Vanrietvelde}
A.\ Vanrietvelde, P.\ A.\ H\"ohn, F.\ Giacomini, and E.\ Castro-Ruiz, \emph{A change of perspective: switching quantum reference frames via a perspective-neutral framework}, Quantum \textbf{4}, 225 (2020).

\bibitem{Hamette}
A.\ de la Hamette and T.\ Galley, \emph{Quantum reference frames for general symmetry groups}, Quantum \textbf{4}, 367 (2020).

\bibitem{Vanrietvelde2}
A.\ Vanrietvelde, P.\ A.\ H\"ohn, and F.\ Giacomini, \emph{Switching quantum reference frames in the N-body problem and the absence of global relational perspectives}, arXiv:1809.05093 [quant-ph].

\bibitem{Hoehn:2018aqt}
P.~A.~H\"ohn and A.~Vanrietvelde, \emph{How to switch between relational quantum clocks}, New J.\ Phys.\ \textbf{22}, 123048 (2020).

\bibitem{Hoehn:2018whn}
P.~A.~H\"ohn, \emph{Switching Internal Times and a New Perspective on the `Wave Function of the Universe'}, Universe \textbf{5}, no.5, 116 (2019).

\bibitem{Hoehn:2019owq}
P.~A.~H\"ohn, A.~R.~H.~Smith and M.~P.~E.~Lock, \emph{The Trinity of Relational Quantum Dynamics}, Phys.\ Rev.\ D \textbf{104}, 066001 (2021).

\bibitem{Hoehn:2020epv}
P.~A.~H\"ohn, A.~R.~H.~Smith and M.~P.~E.~Lock, \emph{Equivalence of approaches to relational quantum dynamics in relativistic settings}, Front.\ in Phys.\ \textbf{9}, 181 (2021).

\bibitem{Hoehn:2021wet}
P.~A.~H\"ohn, M.~P.~E.~Lock, S.~A.~Ahmad, A.~R.~H.~Smith and T.~D.~Galley, \emph{Quantum Relativity of Subsystems}, Phys.\ Rev.\ Lett.\ \textbf{128}, 170401 (2022).

\bibitem{Giacomini:2021gei}
F.~Giacomini, \emph{Spacetime Quantum Reference Frames and superpositions of proper times}, Quantum \textbf{5}, 508 (2021).

\bibitem{Giacomini-spin1}
F.~Giacomini, E.~Castro-Ruiz and \v{C}.~Brukner,
\emph{Relativistic Quantum Reference Frames: The Operational Meaning of Spin}, Phys.\ Rev.\ Lett.\ \textbf{123}, 090404 (2019).

\bibitem{Giacomini-spin2}
L.\ F.\ Streiter, F.\ Giacomini, and \v{C}.\ Brukner, \emph{A Relativistic Bell Test within Quantum Reference Frames}, Phys.\ Rev.\ Lett.\ \textbf{126}, 230403 (2021).

\bibitem{Castro}
E.\ Castro-Ruiz, F.\ Giacomini, A.\ Belenchia, and \v{C}.\ Brukner, \emph{Quantum clocks and the temporal localisability of events in the presence of gravitating quantum systems}, Nat.\ Comm.\ \textbf{11}, 2672 (2020).

\bibitem{Chataignier}
L.~Chataignier,
\emph{Construction of quantum Dirac observables and the emergence of WKB time},
Phys.\ Rev.\ D \textbf{101}, no.8, 086001 (2020).

\bibitem{Chataignier2}
L.~Chataignier,
\emph{Relational observables, reference frames, and conditional probabilities},
Phys.\ Rev.\ D \textbf{103}, no.2, 026013 (2021).

\bibitem{Chataignier3}
L.~Chataignier and M.~Kr\"amer,
\emph{Unitarity of quantum-gravitational corrections to primordial fluctuations in the Born-Oppenheimer approach}
Phys.\ Rev.\ D \textbf{103}, no.6, 066005 (2021).

\bibitem{Ballesteros:2020lgl}
A.~Ballesteros, F.~Giacomini and G.~Gubitosi,
\emph{The group structure of dynamical transformations between quantum reference frames},
Quantum \textbf{5}, 470 (2021).

\bibitem{Mikusch:2021kro}
M.~Mikusch, L.~C.~Barbado and \v{C}.~Brukner,
\emph{Transformation of Spin in Quantum Reference Frames}, Phys.\ Rev.\ Research \textbf{3}, 043138 (2021).


\bibitem{Baumann:2021ifs}
V.~Baumann, M.~Krumm, P.~A.~Gu\'erin and \v{C}.~Brukner, \emph{Noncausal Page-Wootters circuits}, Phys.\ Rev.\ Research \textbf{4}, 013180 (2022).

\bibitem{Savi:2020qdl}
M.~F.~Savi and R.~M.~Angelo,
\emph{Quantum resource covariance},
Phys.\ Rev.\ A \textbf{103}, no.2, 022220 (2021).

\bibitem{Guerin:2018fja}
P.~A.~Gu\'erin and \v{C}.~Brukner,
\emph{Observer-dependent locality of quantum events},
New J.\ Phys.\ \textbf{20}, no.10, 103031 (2018).


\bibitem{Angelo}
R.\ M.\ Angelo, N.\ Brunner, S.\ Popescu, A.\ J.\ Short, and P.\ Skrzypczyk, \emph{Physics within  a quantum reference frame}, J.\ Phys.\ A: Math.\ Theor.\ \textbf{44}, 145304 (2011).


\bibitem{Page}
D.\ N.\ Page, and W.\ K.\ Wootters, \emph{Evolution without evolution: Dynamics described by stationary observables}, Phys.\ Rev.\ \textbf{D 27}, 2885 (1983).

\bibitem{Giovanetti}
V.\ Giovannetti, S.\ Lloyd, and L.\ Maccone, \emph{Quantum time}, Phys.\ Rev.\ D \textbf{79}, 945933 (2015).

\bibitem{Alex1}
A.~R.~H.~Smith and M.~Ahmadi,
\emph{Quantum clocks observe classical and quantum time dilation},
Nat.\ Comm.\ \textbf{11}, no.1, 5360 (2020).

\bibitem{Alex3}
A.\ R.\ H.\ Smith and M.\ Ahmadi, \emph{Quantizing time: interacting clocks and systems},
Quantum \textbf{3}, 160 (2019).


\bibitem{Hardy1}
L.\ Hardy, \emph{The construction interpretation: a conceptual road to quantum gravity}, arXiv:1807.10980.

\bibitem{Hardy2}
L.\ Hardy, \emph{Implementation of the Quantum Equivalence Principle}, in F.\ Finster, D.\ Giulini, J.\ Kleiner, and J.\ Tolksdorf (eds.), Progress and Visions in Quantum Theory in View of Gravity, Birkh\"auser, Cham, 2020.

\bibitem{GiacominiBrukner}
F.\ Giacomini and \v{C}.\ Brukner, \emph{Einstein's Equivalence principle for superpositions of gravitational fields and quantum reference frames}, arXiv:2012.13754.

\bibitem{Conrad}
K.\ Conrad, \emph{Characters of finite Abelian groups}, 2010. \url{https://kconrad.math.uconn.edu/blurbs/grouptheory/charthy.pdf}

\bibitem{periodic}
P.~A.~H\"ohn, M.~P.~E.~Lock, and L.~Chataignier, \emph{Relational dynamics with periodic clock}, to appear.

\bibitem{Simon}
B.\ Simon, \emph{Representations of Finite and Compact Groups}, American Mathematical Society, 1996.

\bibitem{Rovelli}
C.\ Rovelli, \emph{Why Gauge?}, Found.\ Phys.\ \textbf{44}, 91--104 (2014).

\bibitem{Casini:2013rba}
H.~Casini, M.~Huerta and J.~A.~Rosabal,
\emph{Remarks on entanglement entropy for gauge fields},
Phys. Rev. D \textbf{89}, no.8, 085012 (2014).

\bibitem{Donnelly:2016auv}
W.~Donnelly and L.~Freidel,
\emph{Local subsystems in gauge theory and gravity},
JHEP \textbf{09}, 102 (2016).

\bibitem{Geiller:2019bti}
M.~Geiller and P.~Jai-akson,
\emph{Extended actions, dynamics of edge modes, and entanglement entropy},
JHEP \textbf{20}, 134 (2020).

\bibitem{Gomes:2018shn}
H.~Gomes and A.~Riello,
\emph{Unified geometric framework for boundary charges and particle dressings}, 
Phys.\ Rev.\ D \textbf{98}, no.2, 025013 (2018).

\bibitem{Riello:2021lfl}
A.~Riello,
\emph{Edge modes without edge modes},
arXiv:2104.10182 [hep-th].

\bibitem{Wieland:2017zkf}
W.~Wieland,
\emph{New boundary variables for classical and quantum gravity on a null surface},
Class. Quant. Grav. \textbf{34}, no.21, 215008 (2017).

\bibitem{Wieland:2017cmf}
W.~Wieland,
\emph{Fock representation of gravitational boundary modes and the discreteness of the area spectrum},
Annales Henri Poincare \textbf{18}, no.11, 3695-3717 (2017).

\bibitem{Freidel:2020xyx}
L.~Freidel, M.~Geiller and D.~Pranzetti,
\emph{Edge modes of gravity -- I: Corner potentials and charges}, J.\ High Energy Phys.\ \textbf{2020}, 26 (2020).

\bibitem{CH1}
S.~Carrozza and P.~A.~H\"ohn, \emph{Edge modes as reference frames and boundary actions from post-selection}, J.\ High Energy Phys.\ \textbf{2022}, 172 (2022).

\bibitem{all}
A.-C.~de la Hamette, T.~G.~Galley, P.~A.~H\"ohn, L.~Loveridge, and M.~P.~M\"uller, \emph{Perspective-neutral approach to quantum frame covariance for general symmetry groups}, arXiv:2110.13824.

\bibitem{Davidson}
K.\ R.\ Davidson, \emph{C$^*$-Algebras by Example}, American Mathematical Society, 1996.

\bibitem{Savage}
A.\ Savage, \emph{Modern Group Theory}, lecture notes, University of Ottawa, 2017. Available at \url{https://alistairsavage.ca/mat5145/notes/MAT5145-Modern_group_theory.pdf}

\bibitem{LesHouches}
M.\ P.\ M\"uller, \emph{Probabilistic Theories and Reconstructions of Quantum Theory}, SciPost Phys.\ Lect.\ Notes \textbf{28} (2021).

\bibitem{HardyAxioms}
L.\ Hardy, \emph{Quantum Theory From Five Reasonable Axioms}, arXiv:quant-ph/0101012.

\bibitem{DakicBrukner}
B.\ Daki\'c and \v{C}.\ Brukner, \emph{Quantum Theory and beyond: Is entanglement special?}, in ``Deep Beauty. Understanding the Quantum World through Mathematical Innovation'', edited by H.\ Halvorson (Cambridge University Press, New York, 2011).

\bibitem{MasanesMueller}
Ll.\ Masanes and M.\ P.\ M\"uller, \emph{A derivation of quantum theory from physical requirements}, New J.\ Phys.\ \textbf{13}, 063001 (2011).

\bibitem{Chiribella}
G.\ Chiribella, G.\ M.\ D'Ariano, and P.\ Perinotti, \emph{Informational derivation of quantum theory}, Phys.\ Rev.\ A \textbf{84}, 012311 (2011).

\bibitem{HoehnToolbox}
P.\ A.\ H\"ohn, \emph{Toolbox for reconstructing quantum theory from rules on information acquisition}, Quantum \textbf{1}, 38 (2017).

\bibitem{BarnumHilgert}
H.\ Barnum and J.\ Hilgert, \emph{Spectral Properties of Convex Bodies}, J.\ Lie Theory \textbf{30}, 315 (2020).


\end{thebibliography}
\end{document}